\documentclass{article}

 \usepackage{latexsym}
\usepackage{amsmath}
 \usepackage{amsfonts}
\input amssym.def
\input amssym.tex
\usepackage{pdflscape}  
\usepackage[pdftex]{color}
\definecolor{myred}{rgb}{0.6,0,0}

\def\init{\setcounter{equation}{0}}

\newtheorem{theoreme}{Theorem}
\newtheorem{prop}[theoreme]{Proposition}
\newtheorem{rem}[theoreme]{Remark}
\newcommand{\bes}{\begin{subequations}}
\newcommand{\ees}{\end{subequations}}

\def\12{\frac{1}{2}}
\def\sss{{\mathbb S}}
\def\hh{{\mathbb H}}

\def\rr{{\mathbb R}}
\def\zz{{\mathbb Z}}
\def\cc{{\mathbb C}}
\def\nn{{\mathbb N}}

\newcommand\sd{\mathrm{sd}}
\newcommand\F{\mathrm{F}}
\newcommand\Feyn{\mathrm{F}}
\newcommand\aFeyn{\bar{\mathrm{F}}}
\newcommand\PJ{\mathrm{PJ}}

\newcommand\p{{\rm \partial}}

\newcommand\dS{\mathrm{dS}}

\newcommand\loc{{\rm loc}}

\renewcommand\Im{{\rm Im}}
\renewcommand\Re{{\rm Re}}

\newcommand\sing{{\rm sing}}
\newcommand\reg{{\rm reg}}

\newcommand\bep{\begin{proposition}}
\newcommand\eep{\end{proposition}}
\newcommand\ber{\begin{rem}}
\newcommand\eer{\end{rem}}

\newcommand\proof{\noindent {\bf Proof.}\ \ }

\renewcommand\i{{\rm i}}
\newcommand\x{{x}}

\newcommand{\beq}{\begin{equation}}
\newcommand{\eeq}{\end{equation}}
\newcommand{\bear}[1]{\begin{array}{#1}}
\newcommand{\ear}{\end{array}}

\def\otimesal{\mathop{\hbox{\raise 1.5 ex
  \hbox{$\scriptscriptstyle\rm al$}
\kern -0.92 em \hbox{$\otimes$}}}}
\def\oplusal{\mathop{\hbox{\raise 1.5 ex
  \hbox{$\scriptscriptstyle\rm al$}
\kern -0.92 em \hbox{$\oplus$}}}}
\def\Gammal{\hbox{\raise 1.68 ex 
\hbox{$\scriptscriptstyle\rm al$}\kern -0.50 em $\Gamma$}}

\newcommand\sgn{{\rm sgn}}

\newcommand\qed{$\Box$\medskip}

\newcommand\bel{\begin{lemma}}
\newcommand\eel{\end{lemma}}
\newcommand\bet{\begin{theoreme}}
\newcommand\eet{\end{theoreme}}
\newcommand\bex{\begin{example}}
\newcommand\eex{\end{example}}
\newcommand\bed{\begin{definition}}
\newcommand\eed{\end{definition}}
\newcommand\bea{\begin{assumption}}
\newcommand\eea{\end{assumption}}
\newcommand\bec{\begin{corollary}}
\newcommand\eec{\end{corollary}}
\renewcommand\bar{\overline}

\newcommand\supp{{\rm supp}}

\newcommand\e{{\rm e}}

\renewcommand\d{{\rm d}}

\newcommand\cS{{\mathcal S}}

\title{Bessel potentials and Green functions\\on pseudo-Euclidean spaces}
\author{ Jan Derezi\'{n}ski\thanks{ Supported by National Science Center (Poland) under the Grant
UMO-2019/35/B/ST1/01651.} \ and Bart\l{}omiej Sikorski\\  Department of Mathematical Methods in Physics, Faculty of Physics,\\
 University of Warsaw, Pasteura 5, 02-093 Warszawa, Poland\\ email: jan.derezinski@fuw.edu.pl, \	bartlomiej.sikorski@fuw.edu.pl
 }

\begin{document}
\maketitle
\begin{abstract}
   We review properties of Bessel potentials, that is,
inverse  Fourier transforms of (regularizations of)
  $(m^2+p^2)^{-\frac\mu2}$ on a pseudoEuclidean space with
signature $(q,d-q)$. We are mostly interested in  the Lorentzian
signature $(1,d-1)$, and the
case $\mu=2$, related to the Klein-Gordon equation $(-\Box+m^2)f=0$.
We analyze properties
of various ``propagators'', which play an
important role in Quantum Field Theory, such as the retarded/advanced
propagators or Feynman/antiFeynman propagators.
We consistently use  hypergeometric functions instead
of Bessel functions, which makes most formulas much more
transparent. We pay attention to distributional properties of
various Bessel potentials. We include in our analysis the ``tachyonic
case'', corresponding to the ``wrong'' sign in the Klein-Gordon equation.
\end{abstract}


\section{Introduction}
Let us start with the Bessel potentials on the Euclidean space $\rr^d$. 
Let $\Re\mu>0$ and $m\geq0$.
If $m=0$ we will usually additionally assume that $d>\Re\mu$.
Consider the
function
\beq G_{\mu,m}(x)=\int\frac{\e^{\i px}}{(m^2+p^2)^\frac{\mu}{2}}\frac{\d 
  p}{(2\pi)^d}\label{integ}\eeq
on the Euclidean space $\rr^d$.
{    Note that that $G_{\mu,m}(x-y)$ can be interpreted as the integral kernel of 
  the operator $(m^2-\Delta)^{-\frac\mu2}$. }

We have
\beq\label{eq:mass_dependence} G_{\mu,m}(x)= m^{d-\mu}G_{\mu,1}(mx),\eeq
so the case $m>0$ reduces to $m=1$. $G_{\mu,1}(x)$ can be expressed in terms of
the {\em Macdonald function}, one of solutions of the {\em modified
  Bessel equation}. Therefore,  $G_{\mu,1}(x)$
is often called the
{\em Bessel potential} of order $\mu$. 
The function 
$G_{\mu,0}(x)$ is called the {\em Riesz potential} of order $\mu$.

{    It is remarkable that the theory of Bessel potentials is very 
  similar for all $\mu>0$.}
However, the case $\mu=2$ is probably the most important.
In this case we will usually omit $\mu$ from the notation, setting $G_m(x):=G_{2,m}(x)$, and
obtaining the {\em Green function} of the
inhomogeneous {\em Helmholtz equation}
\beq(-\Delta+m^2) g(x)=f(x).\label{helmh}\eeq 
In other words,
\beq(-\Delta+m^2) G_{m}(x)=\delta(x),\eeq

Note that in
dimension $d=3$ we have $G_m(x)=\frac{\e^{-m|x|}}{4\pi|x|}$. Thus
for $m>0$ it coincides with the {\em Yukawa potential} and for $m=0$
with the {\em Coulomb potential}.

Suppose now $\rr^{q,d-q}$ is the {\em pseudo-Euclidean space of
  signature $(q,d-q)$}. In other words, as a set it is $\rr^d$ 
with the scalar product
for $x,y\in\rr^{q,p}$ given by
\beq xy=-x_1y_1\cdots-x_qy_q+x_{q+1}y_{q+1}+\cdots x_{d}y_{d}.\eeq
The definition \eqref{integ} is usually no longer correct {    for $m^2\in\rr$, since
$\frac1{(m^2+p^2)^\frac{\mu}{2}}$ may fail to be locally integrable, and
hence may not define a tempered distribution. It still works for
complex non-real $m^2$.  A possible pair of
generalizations of \eqref{integ} to $m^2$ real is the pair of
functions, which correspond to the limits from above and below:}
\begin{align} G_{\mu,m}^\F(x)=\int\frac{\e^{\i px}}{(m^2+p^2-\i0)^\frac{\mu}{2}}\frac{\d 
  p}{(2\pi)^d},\label{integ1}\\
 G_{\mu,m}^{\bar\F}(x)=\int\frac{\e^{\i px}}{(m^2+p^2+\i0)^\frac{\mu}{2}}\frac{\d 
  p}{(2\pi)^d}.\label{integ2}\end{align}


{    \eqref{integ1} and \eqref{integ2} have an obvious interpretation as 
 boundary values of integral kernels of appropriate functions of the 
 pseudoLaplacian 
 \beq\Box:=-\partial_1^2\cdots-\partial_q^2+\partial_{q+1}^2+\partial_d^2.\eeq}
Again, the case $m>0$ reduces to $m=1$.
 $G_{\mu,m}^{\F/\bar\F}(x)$ can be expressed by Macdonald and Hankel 
 functions. (The Hankel functions are special functions solving 
 the standard Bessel equation.)

 The symbols $\F$ and $\bar\F$ are motivated
 by the special case of Green functions in the Lorentzian case.
 $G_{2,m}^{\F/\bar\F}(x)$ coincide then with 
 the {\em  Feynman}, resp. the {\em anti-Feynman  propagators},
 which play an important role in   Quantum
 Field Theory, as we explain below.

In our paper we will discuss all signatures, including the {\em Euclidean} $(0,d)$ and
{\em anti-Euclidean} $(d,0)$. However, we are mostly interested in the
{\em Lorentzian
signature}. The Lorentzian signature comes in two varieties: ``mostly
pluses'' $(1,d-1)$ and ``mostly minuses'' $(d-1,1)$. We will treat the
former as the standard one.


The Lorentzian case is especially interesting and rich.
This is related to the fact that the  {\em Minkowski space }
$\rr^{1,d-1}$ can be equipped with a causal structure
and the set
$p^2+m^2=0$ has two connected components. Therefore,
{    besides   $G_{\mu,m}^{\F/\bar\F}$, we can introduce}
the distributions 
\begin{align} G_{\mu,m}^\vee(x)=\int\frac{\e^{\i
  px}}{(m^2+p^2-\i0\sgn p^0)^\frac{\mu}{2}}\frac{\d 
  p}{(2\pi)^d},\label{integ1a}\\
 G_{\mu,m}^\wedge(x)=\int\frac{\e^{\i px}}{(m^2+p^2+\i0\sgn p^0)^\frac{\mu}{2}}\frac{\d 
  p}{(2\pi)^d},\label{integ2a}\end{align}
which are invariant wrt orthochronous Lorentz
transformations. Remarkably,  $G_{\mu,m}^{\vee/\wedge}$ is supported
in the forward, resp. backward cone. Therefore,  $G_{\mu,m}^\vee$ is called
the {\em forward (or retarded)}, and  $G_{\mu,m}^\wedge$ the {\em backward (or advanced)  Bessel potential}.

In the Lorentzian case, the  pseudo-Laplacian is usually called the
\emph{d'Alembertian}
\beq 
\Box:=-\partial_0^2+\partial_{1}^2+\cdots+\partial_{d-1}^2,\eeq
and $-\Box+m^2$ is called the {\em Klein-Gordon operator}.
By a  {\em  Green function} of the (inhomogeneous) {\em Klein-Gordon  equation }
\beq(-\Box+m^2) f(x)=g(x).\label{helmh1}\eeq 
we will mean a distribution $G^\bullet(x)$ satisfying 
\beq(-\Box+m^2) G^\bullet(x)=\delta(x).\label{helmh10}\eeq 

The Klein-Gordon equation possesses many Green functions.
Among them, we have the Feynman and antiFeynman Green functions given
by the formulas \eqref{integ1} and \eqref{integ2} with $\mu=2$.
{    Another distinguished pair consists of the retarded (or forward)
Green function and the advanced (or backward) Green function, defined by demanding that their support
is contained in the forward, resp. backward cone.
For $m^2\geq0$ the retarded Green function is given by
\eqref{integ1a} and the advanced Green function by \eqref{integ2a} with $\mu=2$.}

The Feynman, anti-Feynman,
forward, and backward Green functions of the Klein-Gordon equation have important
applications in physics, especially in classical and quantum field
theory. The forward and backward Green functions can be used to
express the Cauchy problem. The Feynman, resp. anti-Feynman Green
functions express the time-ordered, resp. anti-time-ordered vacuum expectation
values of fields in quantum field theory. Importantly, they satisfy the identity
\beq
G_{m}^\F+ G_{m}^{\bar\F}=
G_{m}^{\vee}+ G_{m}^{\wedge}.
\label{spec}\eeq

In our paper, we also consider the Lorentzian case with the ``wrong
sign of $m^2$''. This case 
corresponds to the {\em tachyonic} Klein-Gordon  equation
\beq(-\Box-m^2) f(x)=g(x).\label{helmh2}\eeq 
Remarkably, all four basic Green functions, Feynman $G_m^\F$,
anti-Feynman $G_m^{\bar\F}$,
forward $G_m^\vee$, and backward $G_m^\wedge$,
can be defined in the tachyonic case.
For the Feynman and anti-Feynman Green functions we can still use the formulas
\eqref{integ1} and \eqref{integ2}, where $m^2$ is replaced with
$-m^2$. Their interpretation in terms of the 
vacuum expectation values is however lost, since the tachyonic theory has no
vacuum {    state. (In particular, in the tachyonic case we do not
  have a counterpart of the positive/negative frequency Green
  functions (\ref{pro4})). The forward and backward Green functions
are defined by their support properties. For them} we cannot use the formulas
\eqref{integ1a} and \eqref{integ2a}. In fact, the set $p^2-m^2=0$ is
now connected, and cutting it with $\sgn p^0$ is no longer
invariant. Nevertheless, one can use the
analytic continuation in $m$ to uniquely define Green functions with correct
support properties also in the tachyonic case. 
{   We point out that the
identity \eqref{spec} is no longer true in the tachyonic case. 
}

The difference of two Green functions is a solution of the homogeneous Helmholtz/Klein-Gordon
equation. Certain distinguished solutions are important for
physics applications.  In the Lorentzian case, we have the
Pauli-Jordan propagator; {    for
$m^2\geq0$ also} the positive
frequency and the negative frequency two-point functions.
We illustrate applications of distinguished solutions 
to the Helmholtz/Klein-Gordon equation by computing
averages of plane waves over the sphere (in the
Euclidean case), as well as over the hyperbolic and de Sitter space
(in the Lorentzian case).

Let us say a few words about the history of Bessel potentials.
The name {\em  Bessel potentials} was introduced in the 60s by
Aronszajn and Smith, who studied them in the Euclidean case in \cite{Aronszajn1}. Around the
same time, they were also investigated by Calderon \cite{Calderon}.
  Bessel potentials are frequently viewed in the literature as
  smoothed versions of Riesz potentials (see, for example,
  \cite{Stein} where they are defined using the integral formula
  \eqref{plo1}). {  They are often used to define Bessel potential spaces that generalize standard Sobolev spaces (see \cite{AdamsHedberg}), and the idea to use Bessel kernels is due to Deny \cite{Deny}.} For a comprehensive treatment of (Euclidean) Bessel potentials, we refer the reader to \cite{Aronszajn1}, where many properties of Bessel potentials are exhaustively studied. 

The Lorentzian versions of Bessel potentials, typically in dimension
1+3, often
appear in the literature on Quantum Field Theory.
{    They are ingredients of formulas for scattering amplitudes based on
 Feynman diagrams and on the Epstein-Glaser approach
\cite{Steinmann,microlocal}.} The famous textbooks by
  Björken-Drell    \cite{bjorken} and by Bogoliubov-- Shirkov 
  \cite{bogoliubov} contain  appendices devoted to distinguished
  Green functions and solutions of the Klein-Gordon equation in the
  physical dimension 1+3. They carry various names. For instance,
  often the term {\em Green function} is replaced by {\em propagator}, etc.

{  
Formulas for Bessel potentials in various signatures are known and are available in
collections of integrals such as \cite{BP} and \cite{GR}.  In chapter
III.2 of \cite{GelfandShilov} one can find Fourier transforms of
powers of quadratic forms with any signature, including the formula
(\ref{qre2}) of the general case studied in this paper.
} Although there
exists a large literature about Bessel potentials, our presentation
contains several new points, which we have not seen in the literature
and believe are important.

The first new point involves the special functions that we use.
Various kinds of the Bessel equation can be reduced to equation
\begin{eqnarray}
(z\p_z^2+(\alpha +1)\p_z-1)v(z)=0,
\label{equa1}\end{eqnarray}
 which can be called the
${}_0F_1$ {\em  hypergeometric equation}.
Equation \eqref{equa1} has two singular points: $0$ and $\infty$.
The singularity at $0$ is regular (Fuchsian), and the solution
obtained by the well-known
Frobenius method is the
 ${}_0F_1$ {\em hypergeometric function}, which we denote
 $F_\alpha$. We usually prefer its {\em Olver normalized} version ${\bf 
   F}_\alpha:=\frac{F_\alpha}{\Gamma(\alpha+1)}$, closely related to
 the Bessel function, both standard and modified.

 Another standard
 solution of the ${}_0F_1$   equation, corresponding to the
 irregular singularity at $\infty$, is
 the function that we denote $U_\alpha$,
 This function is perhaps less known. Up to a coefficient, it coincides
 with the {\em Meijer $G$-function} $ G_{0,2}^{2,0}(-;0,-\alpha;z)$.
The function $U_\alpha$ is closely related to the Macdonald
and Hankel functions.

{    In our paper, we treat ${\bf F}_\alpha$ and $U_\alpha$ functions as basic  
elements of our description of Bessel potentials.
In our opinion, they are much more convenient
 for this purpose, rather than functions from the
  Bessel family, as it is done
 in the conventional treatment of this topic.
The corresponding formulas
are simpler and more transparent. This is especially visible when we consider
non-Euclidean signatures, where the formulas involve analytic
continuation across two branches and an irregular distribution at the
junction of these branches. The ${\bf F}_\alpha$ and
$U_\alpha$ functions are also convenient to see the transition from
the Minkowski space to the deSitter and the universal cover of the
AntideSitter space, as discussed in \cite{DeGa}.
 In fact, on the Minkowski space
  retarded/advanced and Feynman/anti-Feynman Bessel potentials
 are expressed in terms of ${\bf F}_\alpha$ and
$U_\alpha$, and on the deSitter and Anti-deSitter space we need closely
related {\em Gegenbauer functions} instead.}

We  also believe that there are some important novel features in our presentation of the
Lorentzian case, which is tailored to the needs of Quantum Field
Theory.
 In our opinion, it is quite remarkable how rich is the theory of Bessel potentials in the Lorentzian signature.
We have four distinct Lorentz invariant Green functions of the
Klein-Gordon equation, with important applications in physics. If we
include also a few useful distinguished solutions to the Klein-Gordon
equation (such as the Pauli-Jordan propagator, positive and negative
frequency solution), then we obtain a whole menagerie of functions.

In
our discussion we cover not only the massive and massless case, but
also the tachyonic case. This case is quite curious, even though usually
ignored in the physics literature. {    We also discuss identity
  \eqref{spec}, true for $m^2\geq0$, but wrong in the tachyonic
  case. Remarkably, this identity sometimes, but not always,
  generalizes to  curved spacetimes, as analyzed recently in \cite{DeGa}.} 

In our treatment, we pay special attention to the
distributional character of Bessel potentials. This is unproblematic
in the Euclidean signature, where Bessel potentials are given by
(locally) integrable functions. This is not the case in non-Euclidean signatures.
In particular, it is interesting to look at the functions ${\bf
  F}_\alpha$ and $U_\alpha$ as defining distributions on the real
line. 
With this interpretation in mind, well-known identities have to be
reformulated, see e.g.
\eqref{conne2a}.

{     Finally, let us mention that there exist a large literature about
Green functions of the Klein-Gordon equation on curved spacetimes. In
the generic context their explicit expression is not possible,
and often instead of exact Green functions one restricts oneself to
{\em parametrices}, that is inverses modulo smoothing terms. The existence
of exactly four parametrices that generalize
$G^{\F/\bar\F}$ and $G^{\vee/\wedge}$ is the result of
a famous paper by Duistermaat and H\"ormander \cite{HormDuist}. It is also
remarkable that expansions similar to
(\ref{qqq1})-(\ref{G_KG_insidecone-}) describe singular parts of these
parametrices also in curved spacetimes, where they can be derived from
the Hadamard recursion relations (see Chapter 4 of \cite{Friedlander}
or Chapter 2 of \cite{WaveEquation}.)  The universality of these
singular  parts is an important idea in
 Quantum Field Theory on curved spacetimes \cite{microlocal}}.

\section{Special functions related to the ${}_0F_1$  equation}
       \init 
\subsection{The ${}_0F_1$  equation}
\label{sa5}
{   
  Our presentation of Bessel potentials will use extensively
  ${}_0F_1$ hypergeometric functions, closely related to functions
  from the Bessel family. Surprisingly, they are seldom
  used and 
   discussed in the literature. Therefore, we devote this
  section
to a concise
   exposition of their properties,    mostly following \cite{De1} and
    \cite{De2}. In particular, we will treat these functions
   as distributions on the real line, as explained in section
   \ref{specialasdistr}, which leads to useful distributional identities 
which we have not seen in the literature.}

Let $c\in\cc$.
The {\em ${}_0F_1$ equation} is
\beq
(z\p_z^2+c\p_z-1)v(z)=0.\label{equa}\eeq
If $c\neq0,-1,-2,\dots$, then the only solution of the ${}_0F_1$ equation equal to $1$ at
$z = 0$
is called
the {\em ${}_0F_1$ hypergeometric function}: 
\[F(c;z):=\sum_{j=0}^\infty
\frac{1}{
(c)_j}\frac{z^j}{j!},\]
{   where $(c)_j$ denotes the Pochhammer symbol:
\begin{align*}
    &(a)_0=1,\\
&(a)_n:=a(a+1)\dots(a+n-1),&&n=1,2,\dots\\
&(a)_n:=\frac{1}{(a-n)\dots(a-1)},&& n=\dots,-2,-1.
\end{align*}}
 $F(c;z)$ is defined for $c\neq0,-1,-2,\dots$.
Sometimes it is more convenient to consider
the function
\[ {\bf F}  (c;z):=\frac{F(c;z)}{\Gamma(c)}=
\sum_{j=0}^\infty
\frac{1}{
\Gamma(c+j)}\frac{z^j}{j!}\]
defined for all $c$.
For all parameters, we have an integral representation called the {\em 
  Schl\"afli formula}: 
\begin{eqnarray*}
\frac{1}{2\pi \i}\int\limits_{]-\infty,0^+,-\infty[}
\e^t\e^{\frac{z}{t}}t^{-c}\d t 
&=& {\bf F}  (c,z),\ \ \ \ \Re z>0,\end{eqnarray*}
where the contour $]-\infty,0^+,-\infty[$ starts at $-\infty$, goes
around $0$ counterclockwise and returns to $-\infty$.

Instead of $c$ it is often more natural to use $\alpha:=c-1$. Thus, we denote
\beq
F_\alpha (z):=F(\alpha+1;z),\quad
 {\bf F}  _\alpha (z):= {\bf F} (\alpha+1;z).
\eeq


The following function is
 also a solution of the ${}_0F_1$ equation \eqref{equa1}:
\begin{eqnarray*}
U_\alpha (z)&:=&\e^{-2\sqrt z} z^{-\frac{\alpha}{2} -\frac14}
{}_2 F_0\Big(\frac12+\alpha,\frac12-\alpha;-;-\frac{1}{4\sqrt z}\Big),
\end{eqnarray*}
where we used the ${}_2F_0$ function, see e.g. \cite{De1,De2}.  $U_\alpha$ is a
multivalued function. When talking about multivalued functions, we will usually consider their {\em principal
  branches}  on the domain $\cc\backslash]-\infty,0]$.

  The function $U_\alpha$ rarely appears in the literature, except as a
  special case of Meijer's function, see \eqref{meijer} below.
  Typically, it is represented through Macdonald or Hankel functions,
  which we describe further in equations \eqref{MacDonald},
  \eqref{Hankel-U}, and \eqref{Hankel+U}. In our opinion, however, the
  function $U_\alpha$ is often more convenient than Macdonald or Hankel functions.

$U_\alpha (z)$ has a symmetry
\begin{eqnarray}\label{U-a_ident}
U_\alpha (z)=z^{-\alpha }U_{-\alpha }(z).
\end{eqnarray}

Alternatively, the function $U_\alpha$ can be defined by the integral
representations valid for all $\alpha$:
\begin{eqnarray}
\frac{1}{\sqrt\pi}\int_0^{\infty}\e^{-t}\e^{-\frac{z}{t}}t^{-\alpha -1}\d t
&=&U_\alpha (z),\ \ \ \ \Re z>0.\label{integU}\end{eqnarray}
For further reference, it is convenient to rewrite
\eqref{integU} as follows:
 For $\Re(m)>0$, we have
\begin{align}
  \int_0^\infty\e^{-tm^2-\frac{x^2}{4t}}t^{-\alpha-1}\d 
  t&=\sqrt\pi m^{2\alpha}U_\alpha\Big(\frac{m^2x^2}{4}\Big)\label{plo1}.
\end{align}
For $\Re(m)\geq0$ \eqref{plo1} is still true in the sense of
oscillatory integrals.
By substituting $x^2\mapsto \e^{\pm\i\frac{\pi}{2}}x^2 , m^2\mapsto \e^{\pm\i\frac{\pi}{2}} m^2,$ into \eqref{plo1} we obtain a pair of identities
valid in terms of  oscillatory integrals for $m>0$:
\begin{align}
  \int_0^\infty\e^{\mp \i tm^2\mp\frac{x^2}{4t}}t^{-\alpha-1}\d 
  t&=\e^{\i\frac{\pi\alpha}{2}}\sqrt\pi m^{2\alpha}U_\alpha\Big(\e^{\pm\i\pi}\frac{m^2x^2}{4}\Big)\label{plo11}.
\end{align}

As $|z|\to\infty$ and 
$    |\arg z|<2\pi-\epsilon$, $\epsilon>0$, we have 
 \beq
U_\alpha (z)\sim{\rm exp}(- 2z^\12) z^{-\frac{\alpha }2-\frac14}.
\label{saddle1}\eeq
$U_\alpha$  is the unique solution of  \eqref{equa1}  with this property.
{    (Note that the validity of \eqref{saddle1} extends beyond
  $|\arg z|<\pi$, that is, beyond the
  principal sheet of the Riemann surface.)}

We can express $U_\alpha$ in terms of the solutions of
with a simple behavior at zero 
\begin{eqnarray}\label{conni}
U_\alpha (z)
&=&\frac{\sqrt\pi}{\sin\pi (-\alpha )} {\bf F}  _\alpha (z)
+\frac{\sqrt \pi}{\sin\pi \alpha }
z^{-\alpha } {\bf F}  _{-\alpha }(z).
\end{eqnarray}
Alternatively, we can use the $U_\alpha$ function and its analytic
continuation around $0$ in the clockwise or anti-clockwise direction
as the basis of solutions:
\begin{equation}
 {\bf F}_\alpha(z)=\frac{\mp\i}{ 2\sqrt\pi}\left(\e^{\mp\i\pi 
  \alpha}
 U _\alpha(z)-
\e^{\pm\i\pi \alpha} U _\alpha(\e^{\pm\i 2\pi}z)\right).\label{conne}
  \end{equation}
 Here is a version of \eqref{conne} adapted to some applications:
  \begin{align}
{\bf F}_\alpha(-z)&=\frac{\i}{ 2\sqrt\pi}\left(\e^{\i\pi 
  \alpha}
 U _\alpha(\e^{\i\pi}z)-
\e^{-\i\pi \alpha} U _\alpha(\e^{-\i
                                 \pi}z)\right),\label{conne1}\\
z^{-\alpha}     {\bf F}_{-\alpha}(-z)&=\frac{\i}{ 2\sqrt\pi}\left(
 U _{\alpha}(\e^{\i\pi}z)-
U _{\alpha}(\e^{-\i \pi}z)\right).\label{conne2}
  \end{align}

We have the recurrence relations
 \begin{align}
 \p_z {\bf F}  _\alpha (z)&= {\bf F}  _{\alpha +1}(z), 
 \\
 \left(z\p_z+\alpha\right) {\bf F}  _\alpha (z)&= {\bf F}  _{\alpha
                                                 -1}(z);\\[2ex]
    \p_z U  _\alpha (z)&=-U _{\alpha +1}(z), \label{recur3}
 \\
 \left(z\p_z+\alpha\right) U  _\alpha (z)&= -U _{\alpha -1}(z).
\end{align}

$\alpha=m\in\zz$ is the degenerate case of the ${}_0F_1$ equation at $0$.
We have then
\[{\bf F}_m
(z)=\sum_{n=\max(0,-m)}\frac{1}{n!(m+n)!}z^n.\]
This easily implies the identity
\begin{eqnarray}
{\bf F}_m(z)&=&z^{-m}
                {\bf F}_{-m}(z).\label{identi0}
\end{eqnarray}
      In the degenerate case $U_\alpha (z)$ needs to be reexpressed
      using the de l'Hospital formula:
\begin{eqnarray}\label{U_integer}
U_m (z)
&=&\frac{(-1)^{m+1}}{\sqrt\pi}\biggl(\sum_{k=1}^{m}\frac{(-1)^{k-1}(k-1)!}{(m-k)!}z^{-k}\\&&+\sum_{j=0}^\infty
                                                                                             \frac{\ln
                                                                                             (z)-\psi(j+m+1)-\psi(j+1)}{j!(m+j)!}z^j
                                                                                             \biggr).
                                                                                             \notag\end{eqnarray}
In the degenerate case, the integral representation simplifies
yielding the so-called {\em Bessel integral representation}.
Besides, we have a generating function:
\begin{eqnarray*}
\frac{1}{2\pi \i}\int\limits_{[0^+]}
\e^{t+\frac{z}t}t^{-m-1}\d t&= &{\bf F}  _m(z)=z^{-m} {\bf F} _{-m}(z),\\
\e^{t}\e^{\frac{z}t}&
=&\sum_{m\in\zz}t^m {\bf F}  _m(z).
\end{eqnarray*} Above,
$[0^+]$ denotes the contour encircling $0$ in the counterclockwise
direction.

In the half-integer case, we can express the ${}_0F_1$ function in terms of
elementary functions. Indeed,
\begin{eqnarray}
F_{-\12}(z)=\cosh2\sqrt z,&&U_{-\12}(z)=\exp(-2\sqrt z),\\
F_{\12}(z)=\frac{\sinh2\sqrt z}{2\sqrt z},&&U_{\12}(z)=\frac{\exp(-2\sqrt z)}{\sqrt z}, 
\end{eqnarray}
and by the recurrence relations, we have for $k\in\nn$
\begin{align}
F_{-\12-k}(z)=& z^{k+\frac{1}{2}} \partial_z^k\Big(
\frac{\cosh(2\sqrt z\big)}{\sqrt{z}}\Big),
\\F_{\12+k}(z)=&\partial_z^k\Big(\frac{\sinh(2\sqrt z)}{2\sqrt z}\Big),\\U_{-\12-k}(z)=&(-1)^kz^{k+\frac{1}{2}} \partial_z^k\Big(\frac{\exp(-2\sqrt z)}{\sqrt{z}}\Big),\\U_{\12+k}(z)=&(-1)^k\partial_z^k\Big(\frac{\exp(-2\sqrt z)}{\sqrt z}\Big).
\end{align}

\subsection{Relationship to confluent functions}

Recall that  the confluent equation is
\beq(w\partial_w^2+(c-w)\partial_w-c)f(w)=0.\eeq
Its standard solutions are
  \begin{align*}
\text{ Kummer's confluent function}\quad  {}_1F_1(a;c;w)&:=\sum_{n=0}^\infty\frac{(a)_n}{(c)_nn!}w^n,\\
\text{and Tricomi's confluent function}\quad
 U(a;c;w)&:=z^{-    a}{}_2F_0(a,1+a-c;-;-w^{-1}).\end{align*} 

The ${}_0F_1$ equation can be reduced to a special class of 
the confluent equation by the so-called {\em Kummer's 2nd transformation}:
\begin{align}
&z\p_z^2+(\alpha+1)\p_z-1\\
=&\frac{4}{w}\e^{-w/2}\Big(w\partial_w^2+(2\alpha+1-w)\partial_w-\alpha-\frac12\Big)\e^{w/2},
\label{gas}\end{align}
where $w=\pm 4\sqrt{z}$, $z=\frac{1}{16}w^2$.
$F_\alpha$ and $U_\alpha$ can be expressed in terms of Kummer's and
Tricomi's confluent function as follows:
\begin{align}F_\alpha(z)&=
\e^{\mp2\sqrt{z}}{}_1F_1\Big(\alpha+\frac12, 
                            2\alpha+1,\pm4\sqrt{z}\Big),\\
    U_\alpha (z)& =\frac{\e^{-2\sqrt z}}{2^{2\alpha+1}}U\Big(\alpha+\frac{1}{2},2\alpha+1,
  4\sqrt{z}\Big).
\end{align}

\subsection{Relationship to
  Meijer G-functions}

Solutions of hypergeometric equations ${}_pF_q$ can be expressed in
terms of Meijer $G$-functions \cite{Mellin-Transform Method}. In particular, the ${}_0F_1$  equation
can be solved by {  two distinguished functions}
\begin{align}
G_{0,2}^{1,0}\Big(\begin{matrix}\\0,-\alpha \end{matrix}\Big|-z\Big)
&:=  \frac{1}{2\pi \i}\int_{L_1} \frac{\Gamma(-s)\e^{\i\pi
                                                                        s}}{\Gamma(\alpha+1+s)}z^s\d s ,\\
    G_{0,2}^{2,0}\Big(\begin{matrix}\\0,-\alpha \end{matrix}\Big|z\Big) &:=\frac{1}{2\pi \i}\int_{L_2} \Gamma(-s)\Gamma(-\alpha-s)z^s\d s.
\end{align}
Here, the contour $L_1$ goes from $+\infty$ to $+\infty$ and encircles
$\nn_0$, and the contour $L_2$ also goes from $+\infty$ to $+\infty$ and encircles
 $\nn_0\cup(\nn_0-\alpha)$, {    both counterclockwise}.
Computing the residues and using the connection formula \eqref{conni}
we obtain
\begin{align}
    {\bf F}_\alpha  (z)&= G_{0,2}^{1,0}\Big(\begin{matrix}\\0,-\alpha \end{matrix}\Big|-z\Big),\\
    U_\alpha(z)&=\frac{1}{\sqrt{\pi}} G_{0,2}^{2,0}\Big(\begin{matrix}\\0,-\alpha \end{matrix}\Big|z\Big).\label{meijer}
\end{align}

\subsection{Relationship  to Bessel functions}

\label{bessel0}
In the literature, the 
 ${}_0F_1$  equation is seldom used. Much more frequent is
the {\em modified Bessel equation}, which  is equivalent to the
${}_0F_1$ equation. {    It is given by the operator}
\begin{eqnarray*}
z^{\frac{\alpha}{2}}\big(z\p_z^2+(\alpha +1)\p_z-1\big)
z^{-\frac{\alpha}{2}}&=&
\partial_w^2+\frac{1}{w}\partial_w-1-\frac{\alpha^2}{w^2},\end{eqnarray*}
where $z=\frac{w^2}{4}$, $w=\pm 2\sqrt z$.

Even more frequent is the (standard) {\em Bessel equation} {    given by}:
\begin{eqnarray*}
-z^{\frac{\alpha}{2}}\big(z\p_z^2+(\alpha +1)\p_z-1\big)
z^{-\frac{\alpha}{2}}&=&
\partial_u^2+\frac{1}{u}\partial_u+1-\frac{\alpha^2}{u^2},
\end{eqnarray*}
where $z=-\frac{u^2}{4}$, $u=\pm 2\i\sqrt z$.
Clearly, we can pass from the modified Bessel to the Bessel equation by
  $w=\pm\i u$.

  The
  function {    ${\bf F}_\alpha$} is also seldom used. Instead,
one uses the {\em modified Bessel function} and, even more frequently,
the {\em Bessel function}:
\begin{eqnarray}
I_\alpha(w)&=&\Big(\frac{w}{2}\Big)^\alpha {\bf F}  _\alpha\Big(\frac{w^2}{4}\Big),\\[3mm]
J_\alpha(w)&=&\Big(\frac{w}{2}\Big)^\alpha {\bf F}  _\alpha\Big(-\frac{w^2}{4}\Big).\label{Bessel-F}
\end{eqnarray}
They solve the modified Bessel, resp. the Bessel equation.

 Instead of the $U_\alpha$ function one uses
the {\em Macdonald function}, solving the modified Bessel equation:
\begin{eqnarray}
K_\alpha(w)&=&\frac{\sqrt\pi}{2}\Big(\frac{
  w}{2}\Big)^\alpha U_\alpha\Big(\frac{w^2}{4}\Big),\label{MacDonald}\end{eqnarray}
 and the Hankel functions of the 1st and 2nd kind, solving the Bessel equation:
\begin{eqnarray}
H_\alpha^{(1)}(w)=H_\alpha^+(w)&=&\frac{- \i}{\sqrt\pi}\Big(\frac{\e^{- \i \pi}
  w}{2}\Big)^\alpha U_\alpha\Big(\e^{-\i \pi}\frac{w^2}{4}\Big),\label{Hankel+U}\\ 
H_\alpha^{(2)}(w)=H_\alpha^-(w)&=&\frac{ \i}{\sqrt\pi}\Big(\frac{\e^{ \i \pi}
  w}{2}\Big)^\alpha U_\alpha\Big(\e^{\i \pi}\frac{w^2}{4}\Big).\label{Hankel-U}\end{eqnarray}
  
Here are the relations between various functions from the Bessel family:
\begin{align}
H_\alpha ^{\pm}(z)&=\frac{2}{\pi}\e^{\mp\i\frac\pi2(\alpha+1)}
K_\alpha (\mp 
\i z),\\
H_{-\alpha}^{\pm}(z)&=\e^{\pm \alpha\pi \i}H_{\alpha}^{\pm}(z)
,\\
J_\alpha (z)&=\frac{1}{2}\left(H_\alpha ^{+}(z)+H_\alpha
              ^{-}(z)\right),\label{ours2}\\
I_\alpha(z)&=\frac{1}{\pi}\bigl(
\mp\i K_\alpha(\e^{\mp\i\pi }z)\pm\i\e^{\i\pi m}K_\alpha(z)\bigr).\label{macdo2}
\end{align}


\subsection{${\bf F}_\alpha$ and $U_\alpha$ functions as distributions}\label{specialasdistr}

The function $U_\alpha(z)$ (and many others that we consider in this
paper) are multivalued analytic functions defined on the Riemann
surface of the logarithm. It has its {\em principal branch} on
$\cc\backslash]-\infty,0]$. For its
analytic continuation around $0$ we will often use the
self-explanatory notation $U_\alpha(\e^{\i\phi}z)$, where
$z\in \cc\backslash]-\infty,0]$ and $\phi\in\rr$.

We will often consider $U_\alpha(w)$ on the real line.
For ${w}>0$ this is
unambiguous. For ${w}<0$ one needs to add $\pm\i0$ indicating whether we
are infinitesimally above or below the real line. At $w=0$ this function has a singularity, which may require a more careful treatment in terms of distributions (see Appendix \ref{Distributions} for notation about some common distributions).

Thus we introduce the distribution on the real line
\beq
U_\alpha({w}\pm\i0):
=\lim_{\epsilon\searrow0}U_\alpha({w}\pm\i\epsilon),\label{podp}\eeq 
where the right-hand side should be understood as the limit in the
distributional sense.
Note that for ${w}\neq0$ these distributions are regular (in the sense of Appendix \ref{Distributions}) and given by
analytic functions:
\begin{align}
  U_\alpha({w}\pm\i0)&=U_\alpha({w}),\quad {w}>0;\\
  U_\alpha({w}\pm\i0)&=U_\alpha\big(\e^{\pm\i\pi}(-{w})\big),\quad {w}<0.
\end{align}
At ${w}=0$ these distributions are irregular if $\Re\alpha\geq1$.
We can then write $ U_\alpha(w\pm\i0)$ as the sum of an irregular and regular part as follows:
\begin{align}\label{decompo}
  U_\alpha(w\pm\i0)&=U_\alpha^\sing(w\pm\i0)+U_\alpha^\reg(w),\\
  U_\alpha^\sing(w\pm\i0)&:=\frac{1}{\sqrt{\pi}}\sum_{j=0}^{\lfloor\Re\alpha\rfloor-1}
                           \frac{(-1)^j\Gamma(\alpha-j) }{j!}(w\pm\i0)^{j-\alpha}.\end{align}
                          This easily follows from \eqref{conni}
                           and \eqref{U_integer}.

                           Recall that for $\alpha\not\in\nn$
                           the symbol $w_-^{-\alpha} $ defined in
                           \eqref{irre}  denotes the standard
regularization of 
$|w|^{-\alpha}\theta(-w)$.
The identity \eqref{conne2} for $w\in\rr\backslash\{0\}$ can be
rewritten as
\begin{align}
w_-^{-\alpha}   {\bf F}_{-\alpha}(w)&:=\frac{\i}{ 2\sqrt\pi}\left(
 U _{\alpha}(w+\i0)-
U _{\alpha}(w-\i0)\right).\label{conne2a}
\end{align}
(Note that both sides of \eqref{conne2a} are zero for $w>0$). It is
easy to see that for $\alpha\not\in\nn$ \eqref{conne2a} is a correct
distributional identity, where the lhs is the product of the distribution
$w_-^{-\alpha} $ and of the smooth function $ {\bf F}_{-\alpha}(w)$,
whereas the rhs is a linear combination of distributions defined in
\eqref{podp}. \eqref{conne2a} can be decomposed into a singular and
regular part as follows:
\begin{align}
  w_-^{-\alpha}    {\bf F}_{-\alpha}(w)&
                                                     =
                                          \sum_{j=0}^{\lfloor\Re\alpha\rfloor-1}\frac{w_-^{-\alpha+j} (-1)^j}{\Gamma(-\alpha+j+1)j!}+
                                          \sum_{j=\lfloor\Re\alpha\rfloor}^\infty
                                          \frac{w_-^{-\alpha+j}(-1)^j}{\Gamma(-\alpha+j+1)j!}
  \label{conne2b}
\end{align}
The rhs of \eqref{conne2a} is well-defined also for $\alpha\in\nn$. We
will {\em define} for such $\alpha$ the symbol on the lhs of
\eqref{conne2a} by the rhs. Using
\eqref{defi3} for $\alpha\in\nn$ we can thus write
\begin{align}
w_-^{-\alpha}    {\bf F}_{-\alpha}(w)&
=
(-1)^{\alpha+1} \sum_{j=0}^{\alpha-1}\frac{(-1)^j\delta^{(\alpha-1-j)}(w)}{j!}+
(-1)^\alpha{\bf F}_\alpha(w)\theta(-w)
                                          .\label{conne2c}
\end{align}
(Compare with \eqref{identi0}, where you do not see the
  distributions supported at zero).

Of course, in the context described in  this subsection, the distribution
$U_\alpha(w\pm\i0)$ defined as in \eqref{podp}
can be also expressed in terms of
  $K_\alpha $ and $H_\alpha ^\pm$,
where we would have to treat $\sqrt{w}$, resp. $\sqrt{-w}$ with
$w\in\rr$ as their arguments.
It is then important to indicate precisely how the analytic continuation of
the square root is performed---whether we bypass the branch point at
zero from above or from below, adding $\pm\i0$ to
the variable:
\bes\begin{align}
 K_\alpha \big(\sqrt{w\mp\i0}\big)&:=\begin{cases}K_\alpha \big(\sqrt{w}\big),& w>0,\\
  K_\alpha (\mp\i \sqrt{-w}\big)=\pm\i\frac{\pi}{2}\e^{\pm\i\pi \alpha}H_\alpha ^\pm\big(\sqrt{-w}\big),&
  w<0;
  \end{cases}\label{compu0a.}\\
  H_\alpha ^\pm\big(\sqrt{-w\pm\i0}\big)&:=\begin{cases}
H_\alpha ^\pm\big(\pm\i\sqrt{w}\big)=
\mp\i\frac{2}{\pi}\e^{\mp\i\pi \alpha}
K_\alpha \big(\sqrt{w}\big),& w>0,\\
  H_\alpha ^\pm\big(\sqrt{-w}\big),&
  w<0.
  \end{cases}\label{compu1a.}
\end{align}\label{comm}\ees
We believe, however, that it is more convenient in such situations
to use the function $U_\alpha$.  Indeed, we have

\begin{align}
   U_\alpha \Big(\frac{w\mp\i0}{4}\Big)&=\begin{cases}\frac{2^{\alpha+1}}{\sqrt\pi}(w\mp\i0)^{-\frac\alpha2}
K_\alpha \big(\sqrt{w\mp\i0}\big)\\
\pm\i 2^{\alpha}\sqrt\pi(w\mp\i0)^{-\frac\alpha2}
H_\alpha ^\pm\big(\sqrt{-w\pm\i0}\big).
  \end{cases}\label{comp.}
\end{align}

\section{Euclidean and anti-Euclidean signature}
       \init 
This section is devoted to Bessel potentials on
the Euclidean space $\rr^d$. $|x|:=\sqrt{x^2}$
will denote the Euclidean norm of $x\in \rr^d$.

In this section, we will provide various expressions both in terms of
the Bessel family functions
$I_\alpha,J_\alpha,K_\alpha,H_\alpha^\pm$, as well as in terms of
the hypergeometric functions $F_\alpha,U_\alpha$.

\subsection{General exponents--Euclidean case}

Consider first the Euclidean signature. For $m>0$ and $\Re\mu>0$ the function
$\frac{1}{(p^2+m^2)^{\frac{\mu}{2}}}$ defines a tempered distribution,
hence one can compute its Fourier transform:
\bet Let $m>0$.
\begin{align}\label{euc}
G_{\mu,m}(x)&= \int\frac{\e^{\i px}}{(p^2+m^2)^{\frac{\mu}{2}}}\frac{\d p}{(2\pi)^d}\\&=
\frac{2}{\Gamma(\frac{\mu}{2})(4\pi)^{\frac{d}{2}}}\Big(\frac{|x|}{2m}\Big)^{\frac{\mu-d}{2}}
K_{\frac{d-\mu}{2}}(m|x|)\\&=
\frac{\sqrt{\pi} m^{d-\mu}}{\Gamma(\frac{\mu}{2})(4\pi)^{\frac{d}{2}}}
U_{\frac{d-\mu}{2}}\Big(\frac{m^2x^2}{4}\Big). \label{porr1}\end{align}
\eet

\proof By (\ref{use1}),
\begin{align}
&\frac{1}{(2\pi)^d}\int\frac{\e^{\i px}\d p}{(m^2+p^2)^{\frac{\mu}{2}}}\\
=&\frac{1}{(2\pi)^d\Gamma(\frac{\mu}{2})}\int_0^\infty\d s\int\d p 
    s^{\frac{\mu}{2}-1}\e^{-(m^2+p^2)s}\e^{\i px}\\
  =&\frac{1}{(4\pi)^{\frac{d}{2}}\Gamma(\frac{\mu}{2})}\int_0^\infty\d s 
    s^{\frac{\mu}{2}-\frac{d}{2}-1}\e^{-m^2s-\frac{x^2}{4s}}
\end{align}
Then we use
          (\ref{plo1}).
\qed

Note that the integrand of \eqref{euc}
is integrable
 for $\Re\mu>d$. Therefore, $G_{\mu,m}$ is bounded for
 such $\mu$.
For instance,
\begin{align}
G_{\mu,m}(0)=  \frac{1}{(2\pi)^d}\int\frac{\d p}{ (p^2+m^2)^{\frac{\mu}{2}}}&=
  \frac{m^{d-\mu}\Gamma(\frac{\mu-d}{2})}{(4\pi)^{\frac{d}{2}}\Gamma(\frac{\mu}{2})},\quad \Re\mu>d.\label{porr}\end{align}

  {  
    \subsection{General exponents--massless case}\label{sec_massless}

For 
$0<\Re\mu<d$ the following function is in $L_\loc^1(\rr^d)$ and is bounded
at infinity,  hence
it defines a regular  distribution in $\cS'(\rr^d)$:
\begin{align}
  G_{\mu,0}(x)&:=
\int\frac{\e^{\i px}}{|p|^\mu}\frac{\d p}{(2\pi)^d}\\&=\frac{\Gamma(\frac{d-\mu}{2})}{\Gamma(\frac{\mu}{2})(4\pi)^{\frac{d}{2}}}\Big(\frac{|x|}{2}\Big)^{\mu-d}.\label{porr2}
  \end{align}
It is called the {\em Riesz potential}, and it is the massless limit of
  Bessel potentials:
  \bet Let
  $0<\Re\mu<d$. Then
  \beq 
G_{\mu,0}(x)=\lim_{m\to0}G_{\mu,m}(x)\label{porr2a}
  \eeq
  in the sense of $\cS'(\rr^d)$.
  \eet

  \proof 
  One can prove this fact in the position space,  see
Subsection \ref{General signature-Zero mass}, where we give a proof in
  the case of a general signature. Instead, in this section we describe a proof
  based on the momentum space.

For $0<\mu<d$, $|p|^{-\mu}$ is a regular distribution.
By using the Dominated Convergence Theorem we see that
the pointwise limit
    \begin{align}\label{eqn:mass_limit}
        \lim_{m\to0} (p^2+m^2)^{-\frac{\mu}{2}}= |p|^{-\mu}
    \end{align}
    is a limit in the  sense of $\cS'(\rr^d)$. The Fourier transformation is a
    continuous operator on $\cS'(\rr^d)$. Therefore, for considered $\mu$,
    \eqref{porr2a} is true. \qed

}

\subsection{General exponent--antiEuclidean case}
 Suppose now the scalar product is negative definite. {    For $m^2>0$, the function
                                                                 $\frac{1}{(-p^2+m^2)^{\frac{\mu}{2}}}$
                                                                 does
                                                                 not
                                                                 define
                                                                 uniquely
                                                                 a
                                                                 distribution,
                                                                 therefore
                                                                 one
                                                                 cannot
                compute its Fourier
    transform.  However, if
      $m^2\in\cc\backslash[0,\infty[$, then
      $\frac{1}{(-p^2+m^2)^{\frac{\mu}{2}}}$ is a tempered
      distribution, and one can take its limit from above or below
in the   distributional sense:}
                                                                 \beq
                                                                 \frac{1}{(-p^2+m^2\pm\i0)^{\frac{\mu}{2}}}:=\lim_{\epsilon\searrow0}
                                                                 \frac{1}{(-p^2+m^2\pm\i\epsilon)^{\frac{\mu}{2}}}.\eeq
{    Thus we obtain 
                                           two kinds of Bessel
  potentials in the antiEuclidean case:}
  \bet
        \begin{align}
          G_{\mu,m}^{\F/\bar \F}(x)&=
\int\frac{\e^{\i px}}{(-p^2+m^2\mp\i0)^{\frac{\mu}{2}}}\frac{\d
                                                        p}{(2\pi)^d}\label{porr1+.}
  \\&=
\frac{\mp\i(\pm\i)^{d}\pi}{\Gamma(\frac{\mu}{2})(4\pi)^{\frac{d}{2}}}\Big(\frac{|x|}{2m}\Big)^{\frac{\mu-d}{2}}
                                                                       H^\mp_{\frac{\mu-d}{2}}(m|x|)
                                                                       \\
  &=
\frac{\mp\i\e^{\pm\i\frac{\pi\mu}{2}}\pi}{\Gamma(\frac{\mu}{2})(4\pi)^{\frac{d}{2}}}\Big(\frac{|x|}{2m}\Big)^{\frac{\mu-d}{2}}
                                                                       H^\mp_{\frac{d-\mu}{2}}(m|x|).\label{porr1++}\\
  &=
\frac{\e^{\pm \i\pi\frac d 2}\sqrt\pi m^{d-\mu}}{\Gamma(\frac{\mu}{2})(4\pi)^{\frac{d}{2}}}
                                                                       U_{\frac{d-\mu}{2}}\Big(\frac{\e^{\pm \i\pi}m^2x^2}{4}\Big).
\label{porr--}  \end{align}
\eet

\proof
Using (\ref{use2}) and then
          (\ref{plo1})
          we obtain
\eqref{porr--}. \qed

Note that the Euclidean Bessel potential
$G_{\mu,m}$ is well defined not only for $m\geq0$, but also for 
$\Re(m)>0$, which guarantees $m^2\in\cc\backslash]-\infty,0]$. 
Taking the limit at the imaginary line
we can express the antiEuclidean Bessel potential in terms of the
Euclidean one:
\beq G_{\mu,m}^{ \F/\bar\F}(x)=\e^{\mp\i\pi\frac{\mu}{2}}G_{\mu,\pm\i m}(x).\eeq
                                                                 
\subsection{Green functions of the Helmholtz equation}

Bessel potentials with $\mu=2$ are
Green functions of the Helmholtz equation
\beq (-E-\Delta)f(x)=g(x).\eeq
More precisely, the Green function for $-E=m^2$ is
\begin{align}
G_{m}(x)&:=\int\frac{\e^{\i px}}{(p^2+m^2)}\frac{\d p}{(2\pi)^d}\\&=
\frac{1}{(2\pi)^{\frac{d}{2}}}\Big(\frac{|x|}{m}\Big)^{1-\frac{d}{2}}
K_{\frac{d}{2}-1}(m|x|)\\&=
 \frac{\sqrt \pi m^{d-2}}{(4\pi)^{\frac{d}{2}}}
U_{\frac{d}{2}-1}\Big(\frac{m^2x^2}{4}\Big).
                                                       \label{porr1,,}
\end{align}
and for $-E=-m^2$ we have two distinguished Green 
functions:
\begin{align}
G_{\mp\i m}(x)
&=
\int\frac{\e^{\i xp}}{(p^2-m^2\mp\i 0)}\frac{\d p}{(2\pi)^d} \\& =\pm\frac{\i}{4}\Big(\frac{m}{2\pi |x|}\Big)^{\frac{d}{2}-1}
  H_{\frac{d}{2}-1}^\pm(m |x|)
\\&=-(-\i)^d\frac{\sqrt \pi m^{d-2}}{(4\pi)^{\frac{d}{2}}}
U_{\frac{d}{2}-1}\Big(-\frac{    m^2(x^2\pm\i0)}{4}\Big).
       \label{besse1}  \end{align}
     $G_{\mp\i m}(x)$ coincide with the case $\mu=2$ of
 the anti-Euclidean Bessel potential
     \eqref{porr1+.} multiplied
by $-1$.

\subsection{Averages of plane waves on sphere}

Consider the sphere in $\rr^d$ of radius $m$, denoted
$\sss_m^{d-1}=\sss_m$. Let $\d\Omega_m$ be the natural
measure on $\sss_m$.
As an application of  Bessel potentials, we will compute the Fourier 
transform of the measure on $\sss_m$.
\bet \label{thm:plane_waves}
\begin{align}
  \int_{\sss_m}\e^{\i p
  x}\d\Omega_m(p)=& 2m^{d-1}\pi^{\frac{d}{2}}{\bf F}_{\frac{d}{2}-1}\Big(-\frac{m^2x^2}{4}\Big) \label{aver1}                  \\=& m^{d-1}(2\pi)^{\frac{d}{2}}(m|x|)^{1-\frac{d}{2}}J_{\frac{d}{2}-1}(m|x|)
.              
\end{align}
\eet
\proof
By the Sochocki-Plemejl formula, we have
\begin{equation}
\delta(|p|-m)=2m\delta(p^2-m^2)=\frac{2m}{2\pi\i}\Big(\frac{1}{p^2-m^2-\i0}-\frac{1}{p^2-m^2+\i0}\Big).
\end{equation}
 Therefore,
\begin{align}
  \int_{\sss_m}\e^{\i p
  x}\Omega_m(p)
  =&\int\e^{\i px}\delta(|p|-m)\d p\\
  =&\frac{2m}{2\pi\i}\int\e^{\i
     px}\Big(\frac{1}{p^2-m^2-\i0}-\frac{1}{p^2-m^2+\i0}\Big)     \d
     p\\
  =&\frac{m(2\pi)^d}{\pi\i}\Big(G_{-\i m}(x)-G_{\i m}(x)\Big)\\
  =
     m^{d-1}\pi^{\frac{d-1}{2}}&\Big((-\i)^{d-1} U_{\frac{d}{2}-1}\Big(\frac{\e^{-\i\pi}m^2x^2}{4}\Big)-\i^{d-1}U_{\frac{d}{2}-1}\Big(\frac{\e^{\i\pi}m^2x^2}{4}\Big)\Big)\\
  =&2m^{d-1}\pi^{\frac{d}{2}}{\bf F}_{\frac{d}{2}-1}\Big(-\frac{m^2x^2}{4}\Big),
\end{align}
where at the end we used \eqref{conne1}.
\qed

Consider a radial function $\rr^d\ni p\mapsto f(|p|)$. Its Fourier
transform is also radial. \eqref{aver1} yields the identity
\begin{align}
\int f(|p|)\e^{-\i px}\d p
=&2\pi^{\frac{d}{2}}\int_0^\infty
f(k){\bf F}_{\frac{d}{2}-1}\Big(-\frac{k^2x^2}{4}\Big) k^{d-1}\d
   k\label{aver5}
  \\
=&(2\pi)^{\frac{d}{2}}\int_0^\infty
f(k)J_{\frac{d}{2}-1}(k|x|)(k|x|)^{-\frac{d}{2}+1} k^{d-1}\d k\label{aver2},
\end{align}
where $k=|p|$ has the meaning of the length of $p$.

Using ${\bf F}_{-\frac12}(-z)=\frac{\cos2\sqrt z}{\sqrt\pi}$ and
${\bf F}_{\frac12}(-z)=\frac{\sin2\sqrt z}{\sqrt{\pi z}}$ we obtain
the low dimensional cases of \eqref{aver5}:
\begin{align}
\int f(|p|)\e^{-\i px}\d p
&=2\int_0^\infty f(k)\cos(k|x|)\d k,\ \ \ \ d=1;\\
&\hspace{-16ex}=2\pi\int_0^\infty
f(k){\bf F}_{0}\Big(-\frac{k^2x^2}{4}\Big) \d k=2\pi\int_0^\infty f(k)kJ_0(k|x|)\d k,\ \ \ \ d=2;\\
&=4\pi\int_0^\infty f(k)k^2\frac{\sin(k|x|)}{k|x|}\d k,\ \ \ \ d=3.
\end{align}

\subsection{Integral representations of the $U_\alpha$ function}

As an illustration of the usefulness of \eqref{aver5}, we will derive
a certain integral represention of $U_\alpha$.

Applying \eqref{aver5} to \eqref{porr1} we obtain
\begin{equation}2\int_0^\infty\frac{k^{d-1}\d k}{(k^2+1)^\frac{\mu}{2}}{\bf
    F}_{\frac{d}{2}-1}\Big(-\frac{r^2
    k^2}{4}\Big)=\frac{\sqrt\pi}{\Gamma(\frac{\mu}{2})}U_{\frac{d-\mu}{2}}\Big(\frac{r^2}{4}\Big). \end{equation} 
Specifying $d=1$ and $d=3$ we obtain
\begin{align}
2\int_0^\infty\frac{\cos(kr)}{(k^2+1)^{\frac{\mu}{2}}}\d k
&=
\frac{\sqrt\pi}{\Gamma(\frac{\mu}{2})}U_{\frac{1-\mu}{2}}\Big(\frac{r^2}{4}\Big),
                                             \label{dedu1-}\\
  4\int_0^\infty\frac{k\sin( kr)}{(k^2+1)^\frac{\mu}{2} r}\d k
&=\frac{\sqrt\pi}{\Gamma(\frac{\mu}{2})}U_{\frac{3-\mu}{2}}\Big(\frac{r^2}{4}\Big).
                                                             \label{dedu3-}\end{align}
                                                   \eqref{dedu3-} could
                                                   be also deduced
                                                   from \eqref{dedu1-}
                                                   by differentiating
                                                   wrt $r$ and using
                                                   the recurrence
                                                   relation \eqref{recur3}.
    Setting $\alpha =\frac{\mu-1}{2}$ in \eqref{dedu1-}, we obtain
    the Poisson
    representation of the $U_\alpha$ function:
\begin{align}
 U_\alpha \Big(\frac{r^2}{4}\Big)&
=\frac{\Gamma(\frac12-\alpha)}{  \sqrt\pi}
  \int_{-\infty}^\infty\e^{-\i kr}(k^2+1)^{\alpha-\frac12}\d k,\ \ \ \alpha <0.\label{poiss1}
\end{align}

\section{General signature}
       \init 

\subsection{Positive mass}
       
Consider now a pseudo-Euclidean space of general signature
$\rr^{q,d-q}$.
{   $  \frac{1}{(p^2+m^2)^{\frac{\mu}{2}}}$  no longer defines
  a tempered distribution in the general signature. Just as in the
  antiEuclidean case, there are two natural regularizations of
  this function:}
 \beq                                                            \frac{1}{(p^2+m^2\pm\i0)^{\frac{\mu}{2}}}:=\lim_{\epsilon\searrow0}
                                                                 \frac{1}{(p^2+m^2\pm\i\epsilon)^{\frac{\mu}{2}}}.\eeq
 They lead to two kinds 
of the Bessel potential:

\bet Let $m>0$ (or more generally $\Re(m)>0$).
Then
\begin{align}
  G_{\mu,m}^{\F/\bar\F}(x)=&\int\frac{\e^{\i px}}{(m^2+p^2\mp\i0)^{\frac{\mu}{2}}}\frac{\d p}{(2\pi)^d}\label{compua}\\
  =&
\frac{2(\pm\i)^q}{\Gamma(\frac{\mu}{2})(4\pi)^{\frac{d}{2}}}\Big(\frac{\sqrt{x^2\pm\i0}}{2m}\Big)^{\frac{\mu-d}{2}}
K_{\frac{d-\mu}{2}}\big(\sqrt{m^2(x^2\pm\i0)}\big)\label{compu0a}\\=&
\mp
\frac{\pi\i(\pm\i)^q}{\Gamma(\frac{\mu}{2})(4\pi)^{\frac{d}{2}}}
\Big(\frac{\sqrt{x^2\pm\i0}}{2m}\Big)^{\frac{\mu-d}{2}}
H_{\frac{\mu-d}{2}}^\mp\big(\sqrt{m^2({-x^2\mp\i0})}\big)
\label{compu1a}
\\
 =& \frac{(\pm\i)^q\sqrt \pi m^{d-\mu}}{\Gamma(\frac{\mu}{2})(4\pi)^{\frac{d}{2}}}
U_{\frac{d-\mu}{2}}\Big(\frac{m^2(x^2\pm\i0)}{4}\Big).\label{qre2}
\end{align}
\label{qwq}\eet

\ber In (\ref{compu0a}) and (\ref{compu1a}) we use the notation explained in
(\ref{compu0a.}) and (\ref{compu1a.}).
Note that (\ref{compu0a}) works  best  for $x^2>0$, because then we can ignore $\pm\i0$. Likewise, 
(\ref{compu1a})  is best suited for $x^2<0$, because then we can ignore $\mp\i0$.

Anyway, in our opinion the expression in terms of $U_\alpha$, \eqref{qre2}, is preferable.
\eer

\noindent{\bf Proof of Thm \ref{qwq}.}
Using (\ref{use2}) and (\ref{use4})
we obtain
  \begin{align}\notag&
    \frac{1}{(2\pi)^d}\int\frac{\e^{\i px}\d p}{(m^2+p^2\mp\i0)^{\frac{\mu}{2}}}\\\notag
=&    \frac{
  \e^{\pm\i\frac{\pi\mu}{4}}}{(2\pi)^d\Gamma(\frac{\mu}{2})}\int_0^\infty\d t\int\d p\e^{\mp\i t(m^2+p^2)}t^{\frac{\mu}{2}-1}\e^{\i px}\\
=&   
    \frac{
      (\pm\i)^q \e^{\pm\i\frac\pi2( \frac{\mu-d}{2})}\pi^{\frac{d}{2}}}{(4\pi)^{\frac{d}{2}}\Gamma(\frac{\mu}{2})}
    \int_0^\infty\d t
    \e^{\mp\i(tm^2-\frac{x^2}{4t})}t^{\frac{\mu-d}{2}-1}.\label{propa}
  \end{align}
  Then we apply \eqref{plo11}.
\qed

{  
\subsection{Zero mass}
\label{General signature-Zero mass}

For $0<\Re\mu<d$
let us introduce two distributions in $\cS'$
\begin{align}
  G_{\mu,0}^{\F/\bar\F}(x)&:=
\int\frac{\e^{\i px}}{(p^2\mp\i0)^{\frac{\mu}{2}}}\frac{\d p}{(2\pi)^d}\label{compua1}\\&=
  \frac{(\pm\i)^q\Gamma(\frac{d-\mu}{2})}{\Gamma(\frac{\mu}{2})(4\pi)^{\frac{d}{2}}}\Big(\frac{x^2\pm\i0}{4}\Big)^{\frac{\mu-d}{2}}.\label{porr2gen}
  \end{align}
They will be called {\em Feynman/antiFeynman Riesz potentials}. They
are massless limits of the corresponding Bessel potentials:

\bet
For  $0<\Re\mu<d$ we have
\beq
  G_{\mu,0}^{\F/\bar\F}(x)=\lim_{m\searrow0}
  G_{\mu,m}^{\F/\bar\F}(x)\eeq
  in the sense of $\cS'$. \label{thm:massless_limit}\eet

  \proof
  Surprisingly, a momentum space proof, from the Euclidean case, seems to be difficult to generalize to the
  non-Euclidean case. Instead, we will present a proof in the
  position space.

Using the decomposition \eqref{decompo} of the function
  $U_\alpha$, we can write
  \begin{align}\label{poij1}
    G_{\mu,m}(x)&=
\frac{(\pm\i)^q}{\Gamma(\frac{\mu}{2})(4\pi)^{\frac{d}{2}}}\Bigg(
                  \sum_{j=0}^{\lfloor\Re\frac{d-\mu}{2}\rfloor-1}
                           \frac{(-1)^j m^j\Gamma(\frac{d-\mu}{2}-j) }{j!}\Big(\frac{(x^2\pm\i0)}{4}\Big)^{j-\frac{d-\mu}{2}}
\\\label{poij2}&+
                   m^{d-\mu}
U_{\frac{d-\mu}{2}}^\reg\Big(\frac{m^2(x^2\pm\i0)}{4}\Big)\Bigg).
\end{align}
The line \eqref{poij1} obviously converges to \eqref{porr2gen}. By 
    \eqref{saddle1}, $U_{\frac{d-\mu}{2}}^\reg$ is a continuous
    function of a polynomial growth at infinity. Therefore, the
    second line \eqref{poij2} converges to zero in $\cS'$. \qed

   Note that as a consequence of the above theorem and of the
   continuity of the Fourier transformation on $\cS'(\rr^d)$ we can
   infer that
   \beq
\lim_{m\searrow 0}   \frac{1}{(p^2+m^2\mp\i0)^{\frac{\mu}{2}}}=
      \frac{1}{(p^2\mp\i0)^{\frac{\mu}{2}}}
      \eeq
      in the sense of $\cS'$.
    }
{\color{magenta}
}

\subsection{Scaling degree of distributions}
Let us start by defining the action of a dilation by $\lambda$ on a distribution $T(x)$ as $T_\lambda(x)=T(\lambda x)$, by which we mean the dual action to the dilation on test functions
\begin{align}
    \langle T_\lambda| f\rangle = \int T(\lambda x)f(x) \d x= \lambda^{-d}\int T(x)f(\lambda^{-1}x) \d x.
\end{align}
Given a distribution $T\in \mathcal{D}'(\rr^{d})$, we define its scaling degree $\mathrm{sd}(T)$ as
\begin{align}\mathrm{sd}(T) =\inf \Big\{\omega: \lim_{\lambda\searrow0} \lambda^{\omega}T_\lambda=0\ \text{ in } \mathcal{D}'(\rr^{d})\Big\}. \end{align}
{   The scaling degree of a distribution is often used in mathematical analysis
of Quantum Field Theory \cite{Steinmann, microlocal}.

Let us compute the scaling degree of Bessel potentials.
\bet 
\beq\sd G_{m,\mu}^{\F/\bar\F}=\begin{cases}d-\mu,&0<\mu\leq d;\\
  0,&d\leq\mu.
  \end{cases}.\eeq
  \eet

  \proof
For $0<\mu<d$,  the Riesz potentials $G_{\mu,0}^{\F/\bar\F}$ defined
  in \eqref{porr2gen} are homogeneous:
\begin{align}
G_{\mu,0}^{\F/\bar\F}(\lambda x)=\lambda^{\mu-d} G_{\mu,0}^{\F/\bar\F}( x).
\end{align}
So $\mathrm{sd} G_{\mu,0}^{\F/\bar\F}=d-\mu$. 

By the definition of the Bessel potential, the mass dependence is \eqref{eq:mass_dependence} 
\begin{align}
    G_{\mu,m}(\lambda x)= \lambda ^{\mu-d}G_{\mu,\lambda m}( x), 
\end{align}
so, according to Theorem \ref{thm:massless_limit},
\begin{align}
    \lim_{\lambda\searrow0}   \lambda ^{d-\mu}G_{\mu,m}(\lambda x)= \lim_{\lambda\searrow0}  G_{\mu,\lambda m}( x)= G_{\mu,0}( x),
\end{align}
which shows that $\mathrm{sd} G_{\mu,m}^{\F/\bar\F}=d-\mu$ for any
mass $m$ and $0<\mu<d$.

For $d<\mu$, $G_m^{
  F/\bar\F}$ is a continuous bounded function, so its scaling degree is
$0$.

For $d=\mu$, we have
\beq
G_{d,m}(x) =\frac{(\pm\i)^q\sqrt \pi m^{d-\mu}}{\Gamma(\frac{d}{2})(4\pi)^{\frac{d}{2}}}
U_0\Big(\frac{m^2(x^2\pm\i0)}{4}\Big).\label{qre20}
\eeq
Now, we can use the bound \eqref{saddle1} and the expansion \eqref{U_integer}
\begin{align}
  |U_0(z\pm\i0)|&\leq C|z|^{-\frac14},\quad z\in\rr,\quad|z|>1;\\
  U_0(z\pm\i0)&=\ln(z\pm\i0)\mathbf{F}_0(z)+H(z),
\end{align}
where $H$ is an entire function, just as $\mathbf{F}_0$. Using this we easily show that
for $\omega>0$
\beq\lambda^\omega G_{d,m}(\lambda x)\to 0\eeq
in the sense of $\cS'$. \qed

}

\section{The Minkowski signature}
       \init 

The Lorentzian signature is especially important, both 
because of its physical relevance and rich mathematical properties.
The spaces
$\rr^{1,d-1}$ and $\rr^{d-1,1}$ are two kinds of a
Minkowski space, that is, a pseudo-Euclidean space with a Lorentzian
signature. We will treat $\rr^{1,d-1}$ as the standard form of a Minkowski space.
 $x^0$ will denote the first coordinate of $\rr^{1,d-1}$, which we assume to be timelike (having a negative coefficient in the scalar product). 
The remaining, spacelike coordinates will be denoted $\vec x$,
so that $x=(x^0,\vec x)$. In other words,
\beq x^2=-(x^0)^2+\vec x^2=-(x^0)^2+(x^1)^2+\cdots+(x^{d-1})^2.\eeq

The future and the past light cone will be denoted 
\begin{align*}
  J^\vee&:=\{x\in\rr^{1,d-1}\ :\quad x^2\leq0,\quad x^0\geq0\},\\
  J^\wedge&:=\{x\in\rr^{1,d-1}\ :\quad x^2\leq0,\quad x^0\leq0\}. 
  \end{align*}

In this section, we will only use the hypergeometric functions ${\bf
  F}_\alpha,U_\alpha$.

\subsection{General exponent}
\label{General exponent}

Let $m>0$.
The set $m^2+p^2$ consists of two connected components: the future and
the past mass hyperboloid.
Therefore, the following four regularizations of
$\frac{1}{(m^2+p^2)^{\frac\mu2}}$ are tempered distributions  invariant wrt the orthochronous Lorentz group:
\beq
\frac{1}{(m^2+p^2\pm\i0)^{\frac{\mu}{2}}},\quad
\frac{1}{(m^2+p^2\pm\i0\sgn p^0)^{\frac{\mu}{2}}}.\eeq
Their inverse Fourier transforms define four kinds of Bessel
potentials:
\begin{align}
  G_{\mu,m}^{\F/\bar\F}(x): =  &\int\frac{\e^{\i px}}{(m^2+p^2\mp\i0)^{\frac{\mu}{2}}}\frac{\d p}{(2\pi)^d}\label{compu-beta}\\
G_{\mu,m}^{\vee/\wedge}(x):=   &\int\frac{\e^{\i px}}{(m^2+p^2\mp\i0\sgn 
                             p^0)^{\frac{\mu}{2}}}\frac{\d 
                             p}{(2\pi)^d}.\label{compux1}\end{align}

By the following well-known  argument, found e.g. in various standard
textbooks on quantum field theory, we can show that
                         $G_{\mu,m}^{\vee/\wedge}$ have causal  supports.

\bet $\supp  G_{\mu,m}^{\vee/\wedge}\subset
J^{\vee/\wedge}$. \label{causal}
\eet

\proof
For definiteness, consider (\ref{compux1}) with the minus sign. In
order to prove that its support is contained in $J^\vee$, by the Lorentz invariance it
suffices to prove that it is zero for $x^0<0$.
We write
\begin{align*}
&\int\frac{\e^{\i px}\d p}{(p^2+m^2-\i 0\sgn p^0)^\frac{\mu}{2}}
  =\int\frac{\e^{-\i p^0x^0+\i \vec p\vec x}
\d
p^0\d\vec p
  }{\big(\vec p^2+m^2-(p^0+\i0)^2\big)^\frac{\mu}{2}}.
\end{align*}
Next, we continuously deform the contour of integration, replacing $ p^0+\i 0$  by $p^0+\i R$, where $R\in[0,\infty[$. We do not cross any singularities of the integrand
and note that $\e^{-\i x^0(p^0+\i
  R)}$ goes to zero (remember that   $x^0<0$).
\qed

\bet\label{thm:Feynman}
We have the identity
\begin{align} \label{iyi}
G_{\mu,m}^\F(x)+G_{\mu,m}^{\bar\F}(x)&=G_{\mu,m}^{\vee}(x)+G_{\mu,m}^{\wedge}(x)
\end{align}
Here are the expressions for the Bessel potentials in the position 
space:
\begin{align}
G_{\mu,m}^{\F/\bar\F}(x)=& \frac{\pm\i\sqrt \pi m^{d-\mu}}{\Gamma(\frac{\mu}{2})(4\pi)^{\frac{d}{2}}}
   U_{\frac{d-\mu}{2}}\Big(\frac{m^2(x^2\pm\i0)}{4}\Big). 
                       \label{compu-beta1}\\
                                                                   G_{\mu,m}^{\vee/\wedge}(x)
= &  \theta(\pm x^0) \frac{2 \pi}{\Gamma(\frac{\mu}{2})(4\pi)^{\frac{d}{2}}}\big(\tfrac{x^2}{4}\big)_-^{\frac{\mu-d}{2}}
{\bf F}_{\frac{\mu-d}{2}}\Big(\frac{m^2x^2}{4}\Big).\label{compux5}
\end{align}
where in \eqref{compux5} we used
the notation introduced in \eqref{conne2a}.

{   
Formula \eqref{compux5} involves the multiplication of a distribution by a 
discontinuous function, which in general is not well defined.
  At the end of this subsection we provide explain how this formula
  can be correctly interpreted.}
\eet 

\proof 
The identity  \eqref{iyi} follows immediately from the defining 
formulas, that is from \eqref{compu-beta} and \eqref{compux1}.

\eqref{compu-beta1}
is a special case of \eqref{qre2}. 
Using \eqref{compu-beta1}  and 
\eqref{compux5} we obtain      
                    a simple expression for the sum of
                     two Bessel potentials:
\begin{align} 
                                        G_{\mu,m}^{\vee}(x)+G_{\mu,m}^{\wedge}(x)=&   \frac{-\i\sqrt \pi m^{d-\mu}}{\Gamma(\frac{\mu}{2})(4\pi)^{\frac{d}{2}}}\Bigg(
U_{\frac{d-\mu}{2}}\Big(\frac{m^2x^2-\i0}{4}\Big)-
                                                                     U_{\frac{d-\mu}{2}}\Big(\frac{m^2x^2+\i0}{4}\Big)\Bigg)\\=&
 \frac{2 \pi}{\Gamma(\frac{\mu}{2})(4\pi)^{\frac{d}{2}}}\big(\tfrac{x^2}{4}\big)_-^{\frac{\mu-d}{2}}
{\bf F}_{\frac{\mu-d}{2}}\Big(\frac{m^2x^2}{4}\Big),\label{pqiu}\end{align}
where again we used
the notation introduced in \eqref{conne2a}.
\eqref{pqiu} is clearly supported in $J^\wedge\cup J^\vee$.
By Thm \ref{causal}, we know that $G_{\mu,m}^{\vee/\wedge}$ are supported
in $J^{\vee/\wedge}$.
Thus to find expressions for $    G_{\mu,m}^{\vee/\wedge}$ we need to
``split the distribution'' \eqref{pqiu} into two terms, one supported
in $J^\vee$ and the other in $J^\wedge$.

Using Proposition \ref{legal} to justify the multiplication of a distribution \eqref{pqiu} by the (discontinuous) function  $\theta(\pm\x^0)$, we can define
         \begin{align}
                                                           \tilde        G_{\mu,m}^{\vee/\wedge}(x)&
=   \theta(\pm x^0) \frac{2 \pi}{\Gamma(\frac{\mu}{2})(4\pi)^{\frac{d}{2}}}\big(\tfrac{x^2}{4}\big)_-^{\frac{\mu-d}{2}}
{\bf F}_{\frac{\mu-d}{2}}\Big(\frac{m^2x^2}{4}\Big).\label{compux5.}
\end{align}
Clearly,  $\tilde G_{\mu,m}^{\vee/\wedge}$ are supported
in $J^{\vee/\wedge}$. Besides,
\beq  \label{sum_of_Gs}   G_{\mu,m}^{\vee}(x)+G_{\mu,m}^{\wedge}(x)=  \tilde
G_{\mu,m}^{\vee}(x)+\tilde G_{\mu,m}^{\wedge}(x).\eeq But
$J^\vee\cap J^\wedge=\{0\}$. Therefore, $G_{\mu,m}^{\vee/\wedge}-\tilde
G_{\mu,m}^{\vee/\wedge}$ is a distribution supported in $\{0\}$, that is, a
linear combination of $\delta^{(\alpha)}(x)$ 
\beq
B_{\mu,m}^{\vee/\wedge}:=G_{\mu,m}^{\vee/\wedge}-\tilde G_{\mu,m}^{\vee/\wedge}
=\sum_{|\alpha|<n}c_{\alpha,m}^{\vee/\wedge}\delta^{(\alpha)}(x).\eeq

(\ref{sum_of_Gs}) implies $B_{\mu,m}^{\vee}(x)=-B_{\mu,m}^{\wedge}(x)$. The symmetry in $x\mapsto-x,\vee/\wedge\mapsto\wedge/\vee$ of (\ref{compux1}) and (\ref{compux5}) allows us to write
$G_{\mu,m}^{\vee}(x)=G_{\mu,m}^{\wedge}(-x)$, $\tilde G_{\mu,m}^{\vee}(x)=\tilde G_{\mu,m}^{\wedge}(-x)$, and therefore $B_{\mu,m}^{\vee/\wedge}(x)=B_{\mu,m}^{\wedge/\vee}(-x)=-B_{\mu,m}^{\vee/\wedge}(-x).$ 
Its action on a test function $\phi\in \cS(\rr^{1,d-1})$ is
\beq \langle B_{\mu,m}^{\vee/\wedge},\phi\rangle=\sum_{|\alpha|<n}(-1)^{|\alpha|}c_{\alpha,m}^{\vee/\wedge}(\partial_{x}^\alpha \phi)(0)=\sum_{|\alpha|<n}c_{\alpha,m}^{\vee/\wedge}(\partial_{-x}^\alpha \phi)(0),\eeq
so $c_{\alpha,m}^{\vee/\wedge}=0$ for even $|\alpha|$.
$G_{\mu,m}^{\vee/\wedge}$ and $\tilde G_{\mu,m}^{\vee/\wedge}$ are
invariant with respect to the action of  the {    proper} Lorentz group. The same must apply to their difference $B_{\mu,m}^{\vee/\wedge}$. Derivatives evaluated at $0$ transform as vectors under the action of the Lorentz group. However, $\langle B_{\mu,m}^{\vee/\wedge},\phi\rangle$ is a sum of terms with only odd number of indices so it cannot be invariant under the action of the Lorentz group
unless $B_{\mu,m}^{\vee/\wedge}=0$. We conclude that
$G_{\mu,m}^{\vee/\wedge}=\tilde
G_{\mu,m}^{\vee/\wedge}$.
\qed

 $G_{\mu,m}^\F$ will be  called the {\em Feynman Bessel potential} and
$G_{\mu,m}^{\bar\F}$  the {\em anti-Feynman Bessel potential}. These names
are somewhat artificial in the context of a general $\mu$. Their
justification comes from the case $\mu=2$, where these Bessel
potentials coincide with 
the Feynman and anti-Feynman propagator known from Quantum Field Theory.

  The distribution $G_{\mu,m}^\vee$ will be  called the {\em forward} or
{\em retarded Bessel potential}, and
$G_{\mu,m}^\wedge$  the {\em backward} or {\em advanced
 Bessel potential}.

{  
For $0<\Re\mu<d$ we also have the massless Riesz potentials:
\begin{align} 
  G_{\mu,0}^{\F/\bar\F}(x)=&\frac{\pm\i\Gamma(\frac{d-\mu}{2})}{\Gamma(\frac{\mu}{2})(4\pi)^{\frac{d}{2}}}\Big(\frac{x^2\pm\i0}{4}\Big)^{\frac{\mu-d}{2}},\label{porr2g}\\
                        G_{\mu,0}^{\vee/\wedge}(x)
= &  \theta(\pm x^0) \frac{2 \pi}{\Gamma(\frac{\mu}{2})\Gamma(\frac{\mu-d+2}{2})(4\pi)^{\frac{d}{2}}}\big(\tfrac{x^2}{4}\big)_-^{\frac{\mu-d}{2}}
.\label{compux50}
  \end{align}
}

{   
As we mentioned above, the formula \eqref{compux5} for the advanced and retarded Bessel
potential involves a product of two distributions, and therefore it
needs a justfication.
We will explain two approaches how to interpret this formula.

The first approach is quite elementary. It
 uses the identification $\rr^{1,d-1}\simeq\rr\times\rr^{d-1}$,
with the first variable denoted $x^0$ or $t$. For the remaining
variables $\vec x$ we will later use spherical coordinates $(r,\Omega)$ with
$r=|\vec x|$.}
  For $n,m\in\nn_0$ and $\chi\in \cS(\rr^{1,d-1})$ let us introduce the 
semi-norms, which involve only the variables $\vec x\in\rr^{d-1}$:
\begin{align*}
    \|\chi(t,\cdot)\|_{n,m} =\sup_{\vec{x}\in
  \mathbb{R}^{d-1},|\alpha|=n,|\beta|=m}=|\vec{x}^{\alpha}(\partial_{\vec x}^{\beta}\chi) (t,\vec{x})|.
\end{align*}

\begin{prop}
\label{legal}
  Let $\text{Re}\nu< d$.
  Then there exist $c_k$, $k=0,\dots, \lfloor\frac{\Re\nu}{2}\rfloor$, such that for any $\phi\in\cS(\rr^{1,d-1})$
  \beq \label{tempineq}
  \Big|\int\big(x^2\big)_-^{-\frac{\nu}{2}}\phi(x)\d x\Big|
  \leq\sum_{k=0}^{\lfloor\frac{\Re\nu}{2}\rfloor}\int c_k
  |t|^{d-\Re\nu+k-1}\|\phi(t,\cdot)\|_{0,k}\d t,\eeq
  where the coefficients $ |t|^{d-\Re\nu+k-1}$ are locally integrable
  and polynomially bounded at infinity.
  Therefore, if $f\in L^\infty(\rr)$, then
  $
f(x^0) \big(x^2\big)_-^{-\frac\nu2}$
defines a tempered distribution on $\rr^d$.  
\end{prop}
\begin{proof}
Action of $\big(x^2\big)_-^{-\frac{\nu}{2}} $ on test a function $\phi \in \mathcal{S}(\rr^d)$ is
\begin{align*}
  \int \big(x^2\big)_-^{-\frac{\nu}{2}}\phi(x)\d x
  =\int_{-\infty}^\infty \d t\int_0^{|t|}\d r \int_{\mathbb{S}^{d-2}}\d \Omega 
(r^2-t^2)_-^{-\frac{\nu}{2}}  \phi(t,r,\Omega) r^{d-2}.
\end{align*}
For simplicity let us consider only $t>0$. We can expand $\phi$ around $r=t$
\begin{align*}
    \phi(t,r,\Omega) =
  \sum_{k=0}^{m}\frac{(r-t)^{k}}{k!}\phi^{(k)}(t,t,\Omega)+
  (r-t)^{m+1}\psi(t,r,\Omega).
\end{align*}
with $m=\lfloor\frac{\text{Re}\nu}{2}\rfloor-1$, where $\phi^{(k)}$ denote derivatives with respect to the $r$ variable. 
Note that $|\psi(t,r,\Omega)|\leq (m+1)!\|\phi(t,\cdot)\|_{0,m+1}$. Let
\begin{align*}
 a_{m+1}:
=&\int_0^\infty \d t\int_0^t\d r \int_{\mathbb{S}^{d-2}}\d \Omega 
(t-r)^{\beta}(r+t)^{-\frac{\nu}{2}} \psi(t,r,\Omega)  r^{d-2},
\end{align*}
with  $\beta=\lfloor\frac{\text{Re}\nu}{2}\rfloor-\frac{\nu}{2}, \quad -1<\text{Re}\beta\leq 0$, is the integral of the locally integrable function.
We see that it is well defined and
\begin{align*}
  & | a_{m+1}|\leq   \int_0^\infty \d t\int_0^t\d r \int_{\mathbb{S}^{d-2}}\d \Omega 
(t-r)^{\text{Re}\beta}(t+r)^{-\text{Re}\frac{\nu}{2}} |\psi(t,r,\Omega)|  r^{d-2}\\
&\leq (m+1)!  \int_0^\infty \d t\|\phi(t,\cdot)\|_{0,m+1} t^{\lfloor\frac{\text{Re}\nu}{2}\rfloor+d-\text{Re}\nu-1}\int_0^1\d r' 
(1-r')^{\beta}(1+r')^{-\frac{\text{Re}\nu}{2}}   r'^{d-2}\int_{\mathbb{S}^{d-2}}\d \Omega \\
&=:  C(d,\nu,m+1) \int_0^\infty \d t t^{\lfloor\frac{\nu}{2}\rfloor+d-\text{Re}\nu-1}\|\phi(t,\cdot)\|_{0,m+1}.
\end{align*}

 Next, we look at each term of the expansion of $\phi(t,r,\Omega)$ in $k$
\begin{align*}
   a_{k}=& \int_0^\infty \d t\int_0^t\d r \int_{\mathbb{S}^{d-2}}\d \Omega 
(r^2-t^2)_-^{-\frac{\nu}{2}} \frac{(r-t)^{k}}{k!}\phi^{(k)}(t,t,\Omega) r^{d-2}\\
=\frac{(-1)^k}{k!}&\int_0^\infty \d t\int_0^t\d r 
(r-t)^{-\frac{\nu}{2}+k} _{-}(t+r)^{-\frac{\nu}{2}} r^{d-2}\int_{\mathbb{S}^{d-2}}\d \Omega \phi^{(k)}(t,t,\Omega).\end{align*}
Here, $(t-r)_-^{-\frac{\nu}{2}+k}$ is the (irregular)  distribution,
defined by (\ref{defi}).
It yields a finite expression:
\begin{align*}
     &\int_0^t\d r 
(r-t)_-^{-\frac{\nu}{2}+k}(t+r)^{-\frac{\nu}{2}} r^{d-2}\\
=t^{d-\nu+k-1}&\int_0^1\d r' 
(r'-1)_-^{-\frac{\nu}{2}+k}(1+r')^{-\frac{\nu}{2}} r^{d-2} =:t^{d-\nu+k-1}\Tilde{C}(d,\nu,k).
\end{align*}

Because $d-\text{Re}\nu+k-1\geq d-\text{Re}\nu -2>-1$, dependence on $t$ is locally integrable and bounded by a polynomial. For $k=0,1,\hdots, m+1$ we can write 
\begin{align*}
   |a_{k}|\leq
C(d,\nu,k) \int_0^\infty \d t \ t^{d-\text{Re}\nu+k-1} \|\phi(t,\cdot)\|_{0,k}.
\end{align*}For fixed $d, \nu$, we have the inequality \eqref{tempineq} showing that homogeneous distributions are tempered distribution. \qed
\end{proof}

{    Now we have $d-\mu<d$, and therefore  Proposition \ref{legal}
 shows
that we can multiply the distribution
$\big(\tfrac{x^2}{4}\big)_-^{\frac{\mu-d}{2}}$
by the 
discontinuous but bounded function
$\theta(\pm x^0)$. The resulting distribution is then multiplied by
the smooth function
${\bf F}_{\frac{\mu-d}{2}}\Big(\frac{m^2x^2}{4}\Big)$, obtaining the
right hand side of \eqref{compux5}.}

{   
  An alternative way to define the product in \eqref{compux5} is 
 based on the concept of the {\em wave front set} \cite{Hormander1}.
  Here are the wave front sets of the distributions contained in
 \eqref{compux5}:
  \begin{align*}
    \mathrm{WF}\left(\theta(t)\right)=&\Big\{\left((0, \Vec{x}),(\tau, 0)\right) \ : \ \Vec{x}\in \mathbb{R}^{d-1}, \tau\neq 0\Big\},\\
    \mathrm{WF}\left(\big(x^2\big)_-^{-\frac{\nu}{2}}\right)=  &\Big\{\left((t, \Vec{x}),(-\lambda t, \lambda\Vec{x})\right) \ : \  t^2-\Vec{x}^2= 0,(t, \Vec{x})\neq 0,\lambda\neq 0\Big\}\\ &\cup\Big\{((0, 0),(\tau, \Vec{k})) \ : \  \tau^2-\Vec{k}^2= 0,(\tau, \Vec{k})\neq 0\Big\},
\end{align*}
 where $(\tau, \Vec{k})$ denotes the dual variable to $(t,
 \Vec{x})$. The fiberwise sum of wavefront sets
 $\mathrm{WF}\left(\theta(t)\right)+\mathrm{WF}\left(\left(x^2\right)_-^{-\frac{\nu}{2}}\right)
 $ does not contain an element of the form $( (t,
 \Vec{x}),(0,0))$. Therefore, by Hörmander's criterion
 \cite[p. 267]{Hormander1}, the product of these two distributions is well defined.
}


\subsection{Green functions of the Klein Gordon equation}

Consider the {\em Klein-Gordon equation}
\beq (-E-\Box)f(x)=g(x),\label{epar}\eeq
where $E$ is a parameter, usually real.  We will consider 3 cases: 
\begin{align} \text{massive case:}\ &-E=m^2,\\
  \text{massless case:}\ &-E=0,\\
  \text{tachyonic  case:}\ &-E=-m^2. 
\end{align}
The massive and massless cases are quite similar and they often appear
in physical applications. {    They are often discussed in detail in
  the  literature. The tachyonic case is more exotic and less known, but
also interesting.}

The  Klein-Gordon equation possesses several useful
Green functions, that is, distributions  satisfying
\beq(-E-\Box)G^\bullet(x)=\delta(x).\eeq
One can try to define
Green functions of the Klein-Gordon equation  the Fourier transformation.
Unfortunately, for $E\in\rr$, $\frac{1}{(-E+p^2)}$ is not a well-defined distribution
because of zeros of its denominator. One way to regularize it is to
add $\pm\i0$, which leads to the so-called Feynman and anti-Feynman
Green function:
\beq G_m^{\F/\bar\F}(x)= \int\frac{\e^{\i px}}{(-E+p^2\mp\i0)}\frac{\d p}{(2\pi)^d}.\label{compum}
\eeq

As follows from a general theory of hyperbolic equations,
the Klein-Gordon equation \eqref{epar} possesses also another important pair of
Green functions: the retarded (or forward) Green function $G^\vee$ and
the advanced (backward) Green function $G^\wedge$.
They are uniquely
defined by the conditions
\beq
\supp G^{\vee/\wedge}\subset J^\pm.\eeq

{    Note that the above definition provides $G^{\vee/\wedge}$ for
  all $E\in\rr$. In the case $-E\geq0$, with $-E=m^2$ they coincide
  with $G_m^{\vee/\wedge}$ defined already with the help of Fourier
  transformation. In the tachyonic case they will be denoted $G_{\i
    m}^{\vee/\wedge}= G_{-\i
    m}^{\vee/\wedge}$
and they  need a separate
  discussion, see Subsection
\ref{Tachyonic Klein-Gordon equation}}

We will also consider {    certain} distinguished
solutions of the (homogeneous) Klein-Gordon equation, that is
functions $G^\circ$ satisfying
\beq(-E-\Box)G^\circ(x)=0.\eeq
One can look for them with the ansatz
\beq G^\circ(x)= \int\e^{\i px}g^\circ(p)\delta(-E+p^2)\frac{\d p}{(2\pi)^{d-1}}\label{compuk},
\eeq
where $g^\circ$ is a distribution on $p^2-E=0$.
Above, for $E\in\rr$, we use the notation
\beq
\delta(p^2-E)\d p=
\frac{\delta\Big(p^0-\sqrt{\vec p^2-E}\Big)}{2\sqrt{\vec p^2-E}}\d \vec p+
\frac{\delta\Big(p^0+\sqrt{\vec p^2-E}\Big)}{2\sqrt{\vec p^2-E}}\d \vec p,
\label{mea2}\eeq
where for $\vec p^2<E$ \eqref{mea2}=0.

Below we consider separately the massive, massless and tachyonic cases
of the Klein-Gordon equation. In all three cases, we will be able
to define 
$ G^{\bar\F/\F}$ and  $G^{\vee/\wedge}.$

\subsection{Massive Klein-Gordon equation}

Let us consider $-E=m^2$, that is the massive Klein-Gordon
equation. The corresponding Green functions satisfy
\beq(m^2-\Box)G_m^\bullet(x)=\delta(x).\eeq
Specifying Theorem \ref{thm:Feynman} to  $\mu =2$, we
obtain the following
expressions for the Feynman and anti-Feynman Green functions:
\bet
\begin{align}
  G_m^{\F/\bar\F}(x)
=  &\int\frac{\e^{\i px}}{(m^2+p^2\mp\i0)}\frac{\d p}{(2\pi)^d}\label{compu}\\
=&  \frac{\pm\i\sqrt \pi m^{d-2}}{(4\pi)^{\frac{d}{2}}}U_{\frac{d}{2}-1}\Big(\frac{m^2(x^2\pm\i 0)}{4}\Big)\label{compuF}.
\end{align}
\eet

The retarded and advanced Green functions of the Klein-Gordon equation
are obtained by specifying Theorem \ref{thm:Feynman}  to $\mu=2$.
 In the following theorem, we also identify their regular and singular part.
\bet
\begin{align}\label{qwqw1}
G_m^{\vee/\wedge}(x)=   &\int\frac{\e^{\i px}}{m^2+p^2\mp\i0\sgn p^0}\frac{\d p}{(2\pi)^d}\\
=&
\theta(\pm x^0)\frac{-\i\sqrt \pi m^{d-2}}{(4\pi)^{\frac{d}{2}}}\Bigg(
U_{\frac{d}{2}-1}\Big(\frac{m^2x^2-\i0}{4}\Big)-
U_{\frac{d}{2}-1}\Big(\frac{m^2x^2+\i0}{4}\Big)\Bigg)\label{qwqw}
\\\label{G_KG_insidecone_anyd}
                      &=\theta(\pm x^0)\frac{2 \pi m^{d-2}}{(4\pi)^{\frac{d}{2}}}\big(\tfrac{x^2}{4}\big)_-^{1-\frac{d}{2}}
                        {\bf F}_{1-\frac{d}{2}}\Big(\frac{m^2x^2}{4}\Big).
\end{align}

We can decompose $ G^{\vee/\wedge}$ into a singular and regular part:
\begin{align}    G_m^{\vee/\wedge} (x)&=G_{m,\sing}^{\vee/\wedge}(x)+
                                     G^{\vee/\wedge}_{m,\mathrm{reg}}(x). \end{align}
                                   For $d$ odd this decomposition can be
                                   chosen as
                                   \begin{align}
               \label{qqq1}                       G_{m,\sing}^{\vee/\wedge}(x)&=
 \frac{ \theta(\pm x^0)}{2\pi^{\frac{d}{2}-1}}
                                                                \sum_{j=0}^{\frac{d-5}{2}}\frac{(-1)^j}{j!\Gamma(2-\frac{d}{2}+j) }\Big(\frac{m^2}{4}\Big)^{j} (x^2)_-^{1-\frac{d}{2}+j},\\
                                  \label{qqq2}                                         G_{m,\reg}^{\vee/\wedge}(x)&=
 \frac{ \theta(\pm x^0)}{2\pi^{\frac{d}{2}-1}} \sum_{j=\frac{d-3}{2}}^{\infty}\frac{(-1)^j}{j!\Gamma(2-\frac{d}{2}+j) }\Big(\frac{m^2}{4}\Big)^{j} (-x^2)^{1-\frac{d}{2}+j}\theta(-x^2).
\end{align}
                                   For $d$ even:
                                   \begin{align}\label{G_KG_insidecone}
 G_{m,\sing}^{\vee/\wedge}(x)                         &=\theta(\pm x^0)\frac{ 1}{2\pi^{\frac{d}{2}-1}}\sum_{j=0}^{\frac{d}{2}-2}\frac{(-1)^{j+1}}{(\frac{d}{2}-2-j)!}\biggl(\frac{m^2}{4} \biggr)^{\frac{d}{2}-2-j}\delta^{(j)}(x^2),\\
    G_{m,\reg}^{\vee/\wedge}(x)&=\theta(\pm x^0)\frac{2\pi m^{d-2}}{(4\pi)^{\frac{d}{2}}}{\bf F}_{\frac{d}{2}-1}\Big(\frac{m^2x^2}{4}\Big)\theta(-x^2). \label{G_KG_insidecone-}
\end{align}
\eet
\proof 
The formula for the Green function of the Klein-Gordon equation
is given by equation \eqref{G_KG_insidecone_anyd}, which was computed
earlier for a general $\mu$ (\ref{compux5}). The decomposition into (\ref{qqq1}) and (\ref{qqq2}) is due to (\ref{conne2b}). For even $d$, the decomposition can be rewritten using (\ref{conne2c}). \qed


Introduce
{    the following} distinguished solutions of the Klein--Gordon equation $-\Box+m^2$:
\begin{align}
  G_m^{\mathrm{PJ}}(x) &:= \frac{\i}{(2\pi )^d} \int\e^{\i x\cdot
                       p}\sgn(p^0)\delta(p^2+m^2) \d p  \label{pro3}\\
&=
\frac{1}{(2\pi)^{d-1}}\int\frac{\d\vec{p}}{\sqrt{\vec p^2+m^2}}\e^{\i \vec x\vec p}
\sin\left( x^0\sqrt{\vec 
    p^2+m^2}\right)
  \\[.5em]
  G_m^{(\pm)}(x) &:= \frac{1}{(2\pi )^d} \int\e^{\i x\cdot
                 p}\theta(\pm p^0)\delta(p^2+m^2) \d p                                   \label{pro4}\\
&=
\frac{1}{(2\pi)^{d-1}}\int\frac{\d\vec{p}}{2\sqrt{\vec p^2+m^2}}\e^{\mp\i x^0\sqrt{\vec 
    p^2+m^2} +\i\vec x\vec p}.
\end{align}
Following \cite{DeGa}, we will call distinguished Green functions and solutions jointly
{   \em propagators}.
$G_m^\PJ$ 
is supported in
 $J^\vee\cup J^\wedge$. Here are the expressions for these solutions
 in terms of positions:
\begin{align}
G_m^{\PJ}(x)
 =& \sgn(x^0) \frac{2 \pi}{(4\pi)^{\frac{d}{2}}}\Big(\frac{x^2}{4}\Big)_-^{\frac{2-d}{2}}
{\bf F}_{\frac{2-d}{2}}\Big(\frac{m^2x^2}{4}\Big) 
,\\\label{propag}
G_m^{(\pm)}(x)=&
\frac{\sqrt{\pi}m^{d-\mu}}{(4\pi)^{\frac{d}{2}}}
U_{\frac{d-2}{2}}\Big(\frac{m^2x^2\pm\i\sgn x^0 0}{4}\Big).
\end{align}

Note the identities satisfied by the {    propagators}:
\begin{subequations}\label{iden}\begin{align}
   G_m^\vee - G_m^\wedge&=G_m^\PJ \label{idi1}\\
        &= \i G_m^{(+)} - \i G_m^{(-)}, \label{idi2}\\
  G_m^{\Feyn} - G_m^{\aFeyn} &= \i G_m^{(+)} +\i G_m^{(-)}, \label{idi3}\\
  G_m^{\Feyn} + G_m^{\aFeyn} &= G_m^\vee + G_m^\wedge, \label{idi4}\\
  G_m^\Feyn &= \i G_m^{(+)} + G_m^\wedge = \i G_m^{(-)} + G_m^\vee, \label{idi5}\\
  G_m^{\aFeyn} &= -\i G_m^{(+)} + G_m^\vee = -\i G_m^{(-)} + G_m^\wedge. \label{idi6}
\end{align}\end{subequations}
To prove these identities we use repeatedly
\beq\theta(\pm p^0)
2\pi\i\delta(p^2+m^2)=
\theta(\pm p^0)\Big(
\frac{1}{p^2+m^2-\i0}-
\frac{1}{p^2+m^2+\i0}\Big),
\label{mea1a}\eeq


\subsection{Massless Klein-Gordon equation}

The massless case is quite similar to the massive one: we need only to
set $m=0$ in the previous subsection. In particular, all identities
\eqref{iden} are satisfied. 
There are a few simplifications.
Only the most singular part of the massive propagator remains in the
massless case. {   This is the special case of Riesz
  potentials, massless limit of Bessel potentials, that we studied in the section \ref{sec_massless}}

\bet
\begin{align}
  G_0^{\F/\bar\F}(x)
=&  \pm\frac{\i\Gamma(\frac{d}{2}-1)}{4\pi^{\frac{d}{2}}}\big(x^2\pm\i
   0\big)^{1-\frac{d}{2}},\\\label{psa1}
   G_0^{\vee/\wedge}(x)=&
\theta(\pm x^0)\frac{1}{2\pi^{\frac{d}{2}-1}\Gamma(2-\frac{d}{2}) 
                              }(x^2)_-^{1-\frac{d}{2}},\\
  \label{psa2}
   G_0^\PJ(x)=&
\sgn(x^0)\frac{1}{2\pi^{\frac{d}{2}-1}\Gamma(2-\frac{d}{2}) 
                              }(x^2)_-^{1-\frac{d}{2}},\\
  G^{(\pm)}_0(x)
=&  \frac{\Gamma(\frac{d}{2}-1)}{4\pi^{\frac{d}{2}}}\big(x^2\pm\i
   0\sgn(x^0)\big)^{1-\frac{d}{2}}.
\end{align}  
For $d$ odd \eqref{psa1} and  \eqref{psa2} can be rewritten as
\begin{align}
      G^{\vee/\wedge}_0(x) 
    &=\theta(\pm x^0)\frac{
      (-1)^{\frac{d}{2}-2}}{2\pi^{\frac{d}{2}-1}}\delta^{(\frac{d}{2}-2)}(x^2),\\
  G^\PJ_0(x) 
    &=\sgn(x^0)\frac{
      (-1)^{\frac{d}{2}-2}}{2\pi^{\frac{d}{2}-1}}\delta^{(\frac{d}{2}-2)}(x^2)
\end{align}
\eet

Note that using (\ref{defi}) we can write 
identity (\ref{idi4}) as
\beq
G^{\F}_0(x)+G^{\bar\F}_0(x) = G_0^\vee + G_0^\wedge=\frac{1}{2\pi^{\frac{d}{2}-1}}\rho_{-}^{\frac{d}{2}-1}(x),
\eeq
which agrees with the fact that massless retarded/advanced Green functions, also known as Riesz distributions (see \cite{WaveEquation}), are expressed by homogeneous distributions supported on $J^{\vee/\wedge}.$

\subsection{Tachyonic Klein-Gordon equation}
\label{Tachyonic Klein-Gordon equation}

Let us now consider the {\em tachyonic Klein-Gordon equation}, which
means, with $E=m^2$. Its Green functions satisfy
\beq(-m^2-\Box)G^\bullet(x)=\delta(x).\eeq
{    Usually, tachyonic quantum fields are considered to be
  unphysical \cite{J}. Nevertheless, every now and then there are attempts to
  analyze them in the physics literature, see \cite{Sudarshan}, and more
  recently \cite{Dragan}.


We have a minor notational problem how to indicate that we replaced
$m^2$ with $-m^2$. Naively, one would think it should be indicated by both $+\i m$ and $-\i m$
  instead of $m$. However, this would suggest the analytic continuation
  $\e^{\i\phi},$ $\pm\phi\in[0,\pi]$, which is not always  appropriate.
  This problem appears in the case of the Feynman propagator:  we will 
  write $G_{\i m}^\F$, but not $G_{-\i m}^\F$. Similarly, for the
  anti-Feynman propagator  we will 
  write $G_{-\i m}^{\bar\F}$, but not $G_{\i m}^{\bar\F}$.
In the case of retarded/advanced propagators, this problem will be
absent, since the analytic continuation can be performed in $m^2$:
thus $G_{\i m}^{\wedge/\vee}=G_{-\i m}^{\wedge/\vee}$.

We define} the Feynman and anti-Feynman Green functions 
by adding $\mp\i0$ to the denominator $-m^2+p^2$ in the momentum
representation. 
In the following theorem, we compute their form in position variables:
\bet
\begin{align}
  G_{\i m}^{\F}(x)/G_{-\i m}^{\bar\F}(x)
  &=\int\frac{\e^{\i px}}{(-m^2+p^2\mp\i0)}\frac{\d p}{(2\pi)^d}\label{compu-ta}\\
  =&\frac{\sqrt \pi m^{d-2}(\mp\i)^{d+1}}{(4\pi)^{\frac{d}{2}}}
U_{\frac{d}{2}-1}\Big(\frac{m^2(-x^2\mp\i0)}{4}\Big) \label{compu-taF}.
\end{align}
In particular, for $x^2>0$ we have
\begin{align*}
    G_{\i m}^{\F}(x)/G_{-\i m}^{\bar\F}(x)=\frac{\pm\i\sqrt \pi m^{d-2}\e^{\mp\i\pi(\frac{d}{2}-1)}}{(4\pi)^{\frac{d}{2}}}U_{\frac{d}{2}-1}\Big(\frac{m^2(-x^2\mp \i0)}{4}\Big)
\end{align*}
and for $x^2<0$
\begin{align*}
   G_{\i m}^{\F}(x)/G_{-\i m}^{\bar\F}(x)=\frac{\pm\i\sqrt \pi m^{d-2}\e^{\mp\i\pi(\frac{d}{2}-1)}}{(4\pi)^{\frac{d}{2}}}U_{\frac{d}{2}-1}\Big(\frac{m^2|x^2|}{4}\Big).
\end{align*}
\eet

\proof Let us start from the usual (positive mass) Feynman
propagator, defined in \eqref{compu} and \eqref{compuF}.
Then we continue analytically 
{    $  G_m^{\F}(x)$ and $G_m^{\bar\F}$, replacing $m$ with $m\e^{\i\phi}$, where
$\phi\in[0,\frac\pi2]$ in the former  and $\phi\in[-\frac\pi2,0]$
in the latter case}. (Note that during the analytic continuation
the denominator has to have a constant sign of its imaginary part,
that is, $\pm\Im(m^2\e^{2\i\phi}+\i0)>0$.)
The analytic continuation yields
\begin{align}
  G^{\F}_{\i m}(x)/G^{\bar\F}_{-\i m}(x)
=&\frac{\pm \i\sqrt \pi m^{d-2}\e^{\pm\i\pi(\frac{d}{2}-1)}}{(4\pi)^{\frac{d}{2}}}
   U_{\frac{d}{2}-1}\Big(\frac{\e^{\pm \i \pi}m^2(x^2\pm\i0)}{4}\Big)\\
  =&\frac{\pm \i\sqrt \pi m^{d-2}\e^{\pm\i\pi(\frac{d}{2}-1)}}{(4\pi)^{\frac{d}{2}}}
U_{\frac{d}{2}-1}\Big(\frac{m^2(-x^2\mp\i0)}{4}\Big),
\end{align}
which coincides with  \eqref{compu-taF}. \qed

{    Unfortunately, the tachyonic 
Feynman and antiFeynman propagator do not have the usual 
physical interpretation, as the vacuum expectation value of the time-ordered, 
resp. anti-time-ordered product of fields. In fact, for tachyons the 
vacuum is ill defined. Nevertheless, some authors, e.g.
\cite{Sudarshan}, try to use  the above Feynman propagator
to define interacting tachyonic quantum
field theory.

Retarded and advanced tachyonic Green functions }
  $G_{\i m}^{\vee/\wedge}$  are not 
tempered distributions on $\rr^{1,d-1}$, and therefore they cannot be 
expressed in terms of the 
Fourier transformation in all variables, as in the massive and
massless  cases {    \eqref{qwqw1}.}  However, they are well-defined, and in the following
theorem we give three equivalent formulas for these propagators.

\bet
The forward and backward propagators in the tachyonic case
are given by
\begin{align}
&G_{\i m}^{\vee/\wedge}(x)   =G_{-\i m}^{\vee/\wedge}(x)\notag\\
=&
\theta(- x^2)\theta(\pm x^0)\frac{\sqrt \pi m^{d-2}\i^{d+1}}{(4\pi)^{\frac{d}{2}}}\Bigg(
U_{\frac{d}{2}-1}\Big(\frac{m^2(-x^2)}{4}\Big)-
   U_{\frac{d}{2}-1}\Big(\frac{\e^{\i2\pi}m^2(-x^2)}{4}\Big)\Bigg)\label{suppi2}\\
  =&
\theta(- x^2)\theta(\pm x^0)\frac{\sqrt \pi m^{d-2}(-\i)^{d+1}}{(4\pi)^{\frac{d}{2}}}\Bigg(
U_{\frac{d}{2}-1}\Big(\frac{\e^{-\i2\pi}m^2(-x^2)}{4}\Big)-
U_{\frac{d}{2}-1}\Big(\frac{m^2(-x^2)}{4}\Big)\Bigg)\label{suppi3}\\
=&\theta(\pm x^0)\frac{2\pi}{(4\pi)^{\frac{d}{2}}}\Big(\frac{x^2}{4}\Big)_-^{1-\frac{d}{2}}{\bf F}_{1-\frac{d}{2}}\Big(\frac{m^2|x^2|}{4}\Big),
\end{align}
They are supported in $J^\vee$, resp. $J^\wedge$.
We can decompose $ G_{\i m}^{\vee/\wedge}$ into a singular and regular part:
\begin{align}    G_{\i m}^{\vee/\wedge} (x)&=G_{\i m,\sing}^{\vee/\wedge}(x)+
                                     G^{\vee/\wedge}_{\i m,\text{reg}}(x). \end{align}
                                   For $d$ odd this decomposition is almost the same as \eqref{qqq1}, \eqref{qqq2} but without the factor $(-1)^j$:
                                   \begin{align}
               \label{qqq1-}                       G_{\i m,\sing}^{\vee/\wedge}(x)&=
 \frac{ \theta(\pm x^0)}{2\pi^{\frac{d}{2}-1}}
                                                                \sum_{j=0}^{\frac{d-5}{2}}\frac{1}{j!\Gamma(2-\frac{d}{2}+j) }\Big(\frac{m^2}{4}\Big)^{j} (x^2)_-^{1-\frac{d}{2}+j},\\
                                  \label{qqq2-}
                                     G_{\i m,\reg}^{\vee/\wedge}(x)&=
 \frac{ \theta(\pm x^0)}{2\pi^{\frac{d}{2}-1}} \sum_{j=\frac{d-3}{2}}^{\infty}\frac{1}{j!\Gamma(2-\frac{d}{2}+j) }\Big(\frac{m^2}{4}\Big)^{j} (-x^2)^{1-\frac{d}{2}+j}\theta(-x^2).
\end{align}
                                   For $d$ even the decomposition is
                                   similar as in
                                   \eqref{G_KG_insidecone} and \eqref{G_KG_insidecone-}: 
                                   \begin{align}\label{G_KG_insidecone+}
 G_{\i m,\sing}^{\vee/\wedge}(x)                         &=\theta(\pm x^0)\frac{1}{2\pi^{\frac{d}{2}-1}}\sum_{j=0}^{\frac{d}{2}-2}\frac{1}{(\frac{d}{2}-2-j)!}\biggl(\frac{m^2}{4} \biggr)^{\frac{d}{2}-2-j}\delta^{(j)}(x^2),\\
    G_{\i m,\reg}^{\vee/\wedge}(x)&=\theta(\pm x^0)\frac{2\pi m^{d-2}}{(4\pi)^{\frac{d}{2}}}{\bf F}_{\frac{d}{2}-1}\Big(\frac{m^2|x^2|}{4}\Big)\theta(-x^2). \label{G_KG_insidecone-+}
\end{align}
 \eet

 \proof Our starting point
 is the formula \eqref{qwqw} for
 the forward and
backward propagator $G_m^{\vee/\wedge}(x)$.
They are  analytic in
$m$.
Therefore, we can apply the analytic continuation 
$m\mapsto \e^{\i\frac{\pi}{2}}m$:
\begin{align}
G _{\e^{\i\frac{\pi}{2}}m}^{\vee/\wedge}(x)
=&\theta(\pm x^0)\frac{-\i\sqrt \pi m^{d-2}\e^{\i\pi(\frac{d}{2}-1)}}{(4\pi)^{\frac{d}{2}}}\Bigg(
U_{\frac{d}{2}-1}\Big(\frac{\e^{\i\pi}(m^2x^2-\i0)}{4}\Big)\notag\\&\qquad-
U_{\frac{d}{2}-1}\Big(\frac{\e^{\i\pi}(m^2x^2+\i0)}{4}\Big)\Bigg).
\end{align}
This yields \eqref{suppi2}.
Alternatively, we can apply
 the analytic continuation 
$m\mapsto \e^{-\i\frac{\pi}{2}}m$, which yields \eqref{suppi3}. \qed

{    Let us
compute the sum of the tachyonic Feynman and antiFeynman propagator:
\begin{align}
    G_{\i m}^{\F}(x)+G_{-\i m}^{\bar\F}(x)
&=\begin{cases}\frac{2 \pi m^{d-2}}{(4\pi)^{\frac{d}{2}}}{\bf
                              F}_{\frac{d}{2}-1}\Big(-\frac{m^2x^2}{4}\Big),&\quad x^2>0;
                            \\
                            \frac{4 \pi}{(4\pi)^{\frac{d}{2}}}\Big(\frac{-x^2}{4}\Big)^{1-\frac{d}{2}}
{\bf F}_{1-\frac{d}{2}}\Big(-\frac{m^2x^2}{4}\Big),&\quad x^2<0.
\end{cases}                             \label{suppi1}\end{align}}
Thus $G_{\i m}^{\F}(x)+G_{-\i m}^{\bar \F}(x)$ does not have a causal support, and consequently,
\beq \label{pumn}
G_{\i m}^\F(x)+G_{-\i m}^{\bar\F}(x)\neq G_{\i m}^{\vee}(x)+G_{\i m}^{\wedge}(x).
\eeq 
The equality in \eqref{pumn} holds only for $x^2<0$. 

\bigskip

Note that because of \eqref{pumn} we could not deduce the formulas of the forward and backward 
propagators from the Feynman and anti-Feynman propagators, and we had
to apply a separate argument based on analytic continuation.

In the tachyonic case, we do not have the solutions
$G_{\i m}^{(\pm)}$. However, we can define the Pauli-Jordan propagator
\begin{align}
G_{\i m}^\PJ(x)=G_{-\i m}^\PJ&=
\frac{1}{(2\pi)^{d-1}}\int\d\vec{p}\,\e^{\i \vec x\vec p}
\frac{\sin\left( x^0\sqrt{\vec 
    p^2-m^2}\right)}{\sqrt{\vec p^2-m^2}}
\\
&=\sgn(x^0)\frac{2\pi}{(4\pi)^{\frac{d}{2}}}\Big(\frac{x^2}{4}\Big)_-^{1-\frac{d}{2}}{\bf F}_{1-\frac{d}{2}}\Big(\frac{m^2|x^2|}{4}\Big).
\end{align}
Note that $G_{\i m}^\PJ
$ cannot be written in the form
\eqref{compuk}.

Among the identities \eqref{iden} only \eqref{idi1} is still true.

\subsection{Averages of plane waves on the hyperbolic plane}

The Minkowski space possesses two kinds of hyperboloids. The
two-sheeted hyperboloid consists of two connected components isomorphic
to the hyperbolic space. In this subsection, we compute the Fourier transform of the natural measure on one 
  of the sheets of the two-sheeted hyperboloids, similarly as in the Theorem \ref{thm:plane_waves}. 

Consider the {\em future/past hyperboloid} in the $d$-dimensional Minkowski space, denoted
$\hh_{\pm,m}=\hh_{\pm,m}^{d-1}$, consisting of points $p$ such that
$p^2+m^2=0$ and $\pm p^0>0$. Let $\d\Omega_m$ denote the standard
measure on $\hh_{\pm,m}$. We will see that up to a coefficient
its Fourier transform is
essentially the ``positive frequency solution of the Klein-Gordon equation.''
\bet
\begin{align}
  \int_{\hh_{\pm,m}}\e^{\i px}\d\Omega_m(p)= m^{d-1}\pi^{\frac{d-1}{2}}
U_{\frac{d}{2}-1}\Big(\frac{m^2(x^2\pm\i\sgn x^0 0)}{4}\Big).
\end{align}
\eet
\proof This average, up to a coefficient, coincides with $G_m^{(\pm)}$
defined in \eqref{pro4},
which we have already computed:
\begin{align}
    \int_{\hh_{\pm,m}}\e^{\i px}\d\Omega_m(p) &=2 m\int \e^{\i px}  \theta(\pm p^0)\delta(p^2 + m^2)\d p\\
                         &=
                           (2\pi)^dmG_m^{(\pm)}(x).
\end{align}
Therefore, it is enough to use the formula
\eqref{propag}. \qed

\subsection{Averages of plane waves on the deSitter space}

The one-sheeted hyperboloid in the physics literature is usually
called the {\em deSitter space}. It will be denoted $\dS_{m}=\dS_{m}^{d-1}$.
It consists of points $p$ such that $p^2=m^2$. Let $\d\Omega_m$ denote
the standard measure on $\dS_{m}$. We will compute the Fourier
transform of the measure on $\dS_{m}$. 
\bet
\begin{align}
  &\int_{\dS_{m}}\e^{\i px}\d\Omega_m(p)=\notag\\\label{wrt}
=    &m^{d-1}\pi^{\frac{d-1}{2}}\Bigg(\i^{d} 
U_{\frac{d}{2}-1}\Big(\frac{m^2(-x^2+\i0)}{4}\Big) + (-\i)^{d} 
U_{\frac{d}{2}-1}\Big(\frac{m^2(-x^2-\i0)}{4}\Big) \Bigg)\\
  =&\begin{cases} (-1)^{\frac{d}{2}}2m^{d-1}\pi^{\frac{d-1}{2}}\Bigg( 
U_{\frac{d}{2}-1}\Big(\frac{-m^2x^2\pm \i 0}{4}\Big) \pm
\sqrt{\pi}\i\big(\tfrac{-x^2}{4}\big)_-^{1-\frac{d}{2}} {\bf F}_{1-\frac{d}{2}}\Big(-\frac{m^2x^2}{4}\Big) \Bigg), &\frac{d}{2}\in \nn,\\ 
2\i ^{d-1} m^{d-1}\pi^{\frac{d}{2}} \big(\tfrac{-x^2}{4}\big)_-^{1-\frac{d}{2}}
                        {\bf F}_{1-\frac{d}{2}}\Big(-\frac{m^2x^2}{4}\Big), &\frac{d}{2}\notin\nn.\end{cases}
\end{align}
\eet
\proof
\begin{align} 
    \int_{\dS_{m}}\e^{\i px}\d\Omega_m(p)& =2 m\int \e^{\i px}  \delta(p^2 - m^2)\d p\\
    &=\frac{m}{\pi \i}\int \e^{\i
      px}\Big(\frac{1}{p^2-m^2-\i0}-\frac{1}{p^2-m^2+\i0}\Big)\d p\\
&=  \frac{m(2\pi)^d}{\pi\i}\big(G_{\i m}^\F(x)-G_{-\i m}^{\bar\F}(x)\big).
\end{align}
Then we can use the result for the tachyonic Feynman and anti-Feynman propagator \eqref{compu-taF}.
 \qed

One can see that the singular part is different in even- and odd-dimensional cases.

        \appendix
\section{Appendix}
\subsection{Some identities}
       \init 
The following identities for $A>0$ follow from the 2nd Euler integral:
\begin{align}
  \frac{1}{A^{\frac{\mu}{2}}}&=\frac{1}{\Gamma(\frac{\mu}{2})}\int_0^\infty\e^{-sA}s^{\frac{\mu}{2}-1}\d s,\label{use1}\\
\frac{1}{(A\pm\i0)^{\frac{\mu}{2}}}&=\frac{\e^{\mp\i\frac{\pi\mu}{4}}}{\Gamma(\frac{\mu}{2})}\int_0^\infty\e^{\pm\i tA}t^{\frac{\mu}{2}-1}\d t.\label{use2}
 \end{align}
We will also need the Fourier transform of the Gaussian function on the Euclidean space $\rr^d$, and of the Fresnel function on the
 pseudo-Euclidean space  $\rr^{q,d-q}$ (with $q$ minuses):
\begin{align}
  \int\d p\e^{-sp^2}\e^{\i px}&=
  \Big(\frac{\pi}{s}\Big)^{\frac{d}{2}}\e^{-\frac{x^2}{4s}},\\
    \int\d p\e^{\pm\i tp^2}\e^{\i px}&=
 (\mp \i)^q\Big(\frac{\pi}{t}\Big)^{\frac{d}{2}}\e^{\pm\i\frac{\pi}{4}d}\e^{\mp\i\frac{x^2}{4t}}\label{use4}.
\end{align}

\subsection{Distributions}
\label{Distributions}

In this paper, we often use the language of distributions on $\rr^d$.
We say that a distribution $T$ is {\em regular} if there exists a
locally integrable function $f$ such that for a test function $\Phi$
\beq T(\Phi)=\int f(x)\Phi(x)\d x.\eeq
We will use the integral notation also for irregular distributions,
e.g.
\beq \int\delta^{(n)}(x)\Phi(x)\d x=(-1)^n\Phi^{(n)}(0).\eeq

Let us now consider some special distributions on $\rr$.
For any $\lambda\in\cc$
\[(\pm\i x+0)^\lambda=\e^{\pm\i\lambda\frac{\pi}{2}}(x\mp \i0)^\lambda:=\lim_{\epsilon\searrow0}(\pm\i x+\epsilon)^\lambda.\]
is a tempered distribution. If $\Re\lambda>-1$, then it is regular and
 given by the locally integrable function
\beq\e^{\pm\i\sgn(x)\frac{\pi}{2}\lambda}|x|^\lambda.\eeq

The functions
\beq
x_\pm^\lambda:=|x|^\lambda\theta(\pm x)\label{homog}\eeq
define regular distributions only for $\Re\lambda>-1$. We can extend them
to $\lambda\in\cc$ except for $\lambda=-1,-2,\dots$ by putting
\begin{align}
x_\pm^\lambda&:=\frac{1}{2\i\sin\pi\lambda}
\Big(-\e^{-\i\frac{\pi}{2}\lambda}(\mp\i x+0)^\lambda
+\e^{\i\frac{\pi}{2}\lambda}(\pm\i x+0)^\lambda\Big).
\label{irre}\end{align}
For $\lambda>-1$ \eqref{irre}  are regular and coincide with
$\theta(\pm x)|x|^\lambda$.
We have
\beq
x_\pm^{\lambda+1}=|x|\cdot x_\pm^\lambda.\eeq

Instead of $x_\pm^\lambda$, it is often more convenient to consider
\begin{align}
\rho_\pm^\lambda(x)&:=\frac{x_\pm^\lambda}{\Gamma(\lambda+1)}\label{homog1}
\\
&=\frac{\Gamma(-\lambda)}{2\pi\i}
\Big(\e^{-\i\frac{\pi}{2}\lambda}(\mp\i x+0)^\lambda
-\e^{\i\frac{\pi}{2}\lambda}(\pm\i x+0)^\lambda\Big).
\label{defi}\end{align}
Note that using (\ref{homog1}) and  (\ref{defi}) we have defined $\rho_\pm^\lambda$ for all $\lambda\in\cc$.
We have
 \[\p_x\rho_\pm^\lambda(x)=\pm\rho_\pm^{\lambda-1}(x).\]
 At  integers we have
 \begin{align}  \rho_\pm^n(x)&=\frac{x_\pm^n}{n!},\ \ \ n=0,1,\dots;\\
   \rho_\pm^{-n-1}(x)&=(\pm1)^{n}\delta^{n}(x),\ \ \ n=0,1,\dots.
 \label{defi3}\end{align}
Clearly, for $\Re(\lambda)\leq-1$ the distributions $\rho_\pm^\lambda$
are irregular.

\paragraph{Acknowledgement.} 
The support of the National Science Center of Poland under the 
    grant UMO-2019/35/B/ST1/01651 is acknowledged.
{    We thank the referees for their remarks. In particular, we are
grateful to one of them for drawing our attention to H\"ormander's
criterion for multiplication of distribution; see its application
at the end of Subsection \ref{General exponent}}

\end{document}